\documentclass[letterpaper,11pt]{article}
\usepackage[utf8]{inputenc}

\usepackage{amsmath, amsthm, amssymb, thm-restate}
\usepackage{algorithmicx}
\usepackage[table,xcdraw]{xcolor}
\usepackage[ruled,vlined,linesnumbered]{algorithm2e}
\usepackage{setspace}
\usepackage{mathtools}
\usepackage[numbers]{natbib}
\usepackage{comment} 
\usepackage[most]{tcolorbox}   
\usepackage{xfrac}
\usepackage{hyperref}
\usepackage{multirow}
\usepackage{caption}
\usepackage{bm}
\usepackage{newfloat}
\usepackage{enumitem}
\usepackage{wrapfig}
\usepackage{dsfont}
\usepackage{multirow}
\usepackage{etoolbox}          
\usepackage{fancyhdr}
\usepackage{csquotes}
\usepackage[noabbrev]{cleveref}
\usepackage{tabularx}
\usepackage[table]{xcolor} 
\usepackage{esvect}

\usepackage[margin=1in]{geometry}

\allowdisplaybreaks

\definecolor{mygreen}{RGB}{10,110,230}
\definecolor{myred}{RGB}{10,110,230}

\hypersetup{
     colorlinks=true,
     citecolor=mygreen,
     linkcolor=myred
}

\renewcommand{\epsilon}{\varepsilon}

\DeclareMathOperator{\E}{\ensuremath{\normalfont \textbf{E}}}

\newcounter{shanecommentcount}

\newcommand{\hiddencomment}[1]{}

\newcommand{\mc}[1]{\ensuremath{\mathcal{#1}}}

\newcounter{protocolcounter}
 
\crefname{protocolcounter}{Algorithm}{Algorithms}

\newcommand{\opt}[0]{\mathrm{OPT}}

\newcommand{\yin}{y_{in}}
\newcommand{\yout}{y_{out}}

\newcommand{\IL}{E_{in}^{lo}}
\newcommand{\OL}{E_{out}^{lo}}
\newcommand{\IH}{E_{in}^{hi}}
\newcommand{\OH}{E_{out}^{hi}}

\newcommand{\zli}{z_{in}^{lo}}
\newcommand{\zlo}{z_{out}^{lo}}
\newcommand{\zliest}{Z_{in}^{lo}}
\newcommand{\zloest}{Z_{out}^{lo}}

\newcommand{\pos}{\operatorname{pos}}
\newcommand{\Pest}{\operatorname{P}}
\newcommand{\clamp}[1]{\operatorname{clamp}\paren*{#1}}
\newcommand{\altclamp}{\operatorname{clamp}}

\newcommand{\val}[2]{\text{val}_{#1}\paren*{#2}}
\newcommand{\maxval}[1]{\operatorname{maxval}_{#1}}

\newcommand{\indeg}{d^{-}}
\newcommand{\outdeg}{d^{+}}
\newcommand{\B}{\textbf{B}}
\newcommand{\C}{\textbf{C}}

\newcommand{\tpos}{\widetilde{\pos}}

\newcommand{\zb}{\overline{z}}

\newcommand{\VertexEstimator}{\mathsf{VertexEstimator}}
\newcommand{\Neighbors}{\textbf{N}}
\newcommand{\HTAvg}{\mathsf{HTAvg}}
\newcommand{\EdgeEstimator}{\mathsf{EdgeEstimator}}
\newcommand{\Rsvsin}{R_{in}}
\newcommand{\Rsvsout}{R_{out}}
\newcommand{\zbest}{\overline{Z}}
\newcommand{\degB}{d_{\B}}
\newcommand{\outdegB}{d^+_{\B}}

\newcommand{\cutval}{\textbf{Cut-Val}}
\newcommand{\T}{\textbf{T}}
\newcommand{\Ex}{\mathbb{E}}
\newcommand{\Tmax}{T_{max}}

\newcommand{\Var}{\operatorname{Var}}
\newcommand{\Cov}{\operatorname{Cov}}
\newcommand{\Out}{\mathsf{Out}}

\newcommand{\card}[1]{|#1|}

\newcommand{\zliesth}{\mc{Z}_{in}^{lo}}
\newcommand{\zloesth}{\mc{Z}_{out}^{lo}}
\newcommand{\Pesth}{\operatorname{\mc{P}}}
\newcommand{\zbesth}{\overline{\mc{Z}}}

\SetKwRepeat{Do}{do}{while}

\DeclarePairedDelimiter{\paren}{\lparen}{\rparen}
\DeclarePairedDelimiter{\set}{\lbrace}{\rbrace}

\DeclarePairedDelimiter{\bracket}{[}{]}
\DeclarePairedDelimiter{\abs}{\lvert}{\rvert}

\DeclareMathOperator{\poly}{poly}

\renewcommand{\O}[1]{\ensuremath{O\left(#1\right)}}

\renewcommand{\epsilon}[0]{\ensuremath{\varepsilon}}

\let\originalleft\left
\let\originalright\right
\renewcommand{\left}{\mathopen{}\mathclose\bgroup\originalleft}
\renewcommand{\right}{\aftergroup\egroup\originalright}

\newtheorem{theorem}{Theorem}

\newtheorem{lemma}{Lemma}[section]

\newtheorem{proposition}[lemma]{Proposition}

\newtheorem{definition}[lemma]{Definition}
\newtheorem{claim}[lemma]{Claim}

\makeatletter
\def\thm@space@setup{%
  \thm@preskip= 0.2cm
  \thm@postskip=\thm@preskip 
}
\makeatother

\crefname{lemma}{Lemma}{Lemmas}
\crefname{theorem}{Theorem}{Theorems}
\crefname{property}{Property}{Properties}
\crefname{claim}{Claim}{Claims}
\crefname{corollary}{Corollary}{Corollaries}
\crefname{result}{Result}{Results}
\crefname{conj}{Conjecture}{Conjectures}
\crefname{definition}{Definition}{Definitions}
\crefname{observation}{Observation}{Observations}
\crefname{proposition}{Proposition}{Propositions}
\crefname{assumption}{Assumption}{Assumptions}
\crefname{line}{Line}{Lines}
\crefname{figure}{Figure}{Figures}
\creflabelformat{property}{(#1)#2#3}
\crefname{equation}{}{}
\crefname{section}{Section}{Sections}
\crefname{appendix}{Appendix}{Appendices}
\crefname{problem}{Problem}{Problems}
\crefname{algorithm}{Algorithm}{Algorithms}

\definecolor{mygreen}{RGB}{20,155,20}
\definecolor{myred}{RGB}{195,20,20}
\definecolor{linkcolor}{RGB}{0,0,230}
\definecolor{mylightgray}{RGB}{230,230,230}
\definecolor{verylightgray}{RGB}{240,240,240}
\definecolor{commentcolor}{RGB}{120,120,120}

\SetCommentSty{mycommfont}

\hypersetup{
     colorlinks=true,
     citecolor= mygreen,
     linkcolor= myred,
     urlcolor= mygreen
}

\renewcommand{\mc}[1]{\ensuremath{\mathcal{#1}}}

\newcounter{myalgctr}

\newenvironment{tbox}{
\par\addvspace{0.2cm}
\begin{tcolorbox}[width=\textwidth,
                  boxsep=2pt,
                  left=1pt,
                  right=1pt,
                  top=4pt,
                  boxrule=1pt,
                  arc=0pt,
                  colback=white,
                  colframe=black
                  ]
}{
\end{tcolorbox}
}

\newenvironment{tboxh}{
\par\addvspace{0.2cm}
\begin{tcolorbox}[width=\textwidth,
                  boxsep=2pt,
                  left=1pt,
                  right=1pt,
                  top=4pt,
                  boxrule=1pt,
                  arc=0pt,
                  colback=white,
                  colframe=black,
                  float=t
                  ]
}{
\end{tcolorbox}
}

\newenvironment{graytbox}{
\par\addvspace{0.1cm}
\begin{tcolorbox}[width=\textwidth,
                  frame hidden,
                  boxsep=5pt,
                  left=1pt,
                  right=1pt,
                  top=2pt,
                  bottom=2pt,
                  boxrule=1pt,
                  arc=0pt,
                  colback=mylightgray,
                  colframe=black,
                  breakable
                  ]
}{
\end{tcolorbox}
}

\newcommand{\tboxhrule}[0]{\vspace{0.1cm} \hrule \vspace{0.2cm}}

\newenvironment{titledtbox}[1]{\begin{tbox}#1 \tboxhrule}{\end{tbox}}
\newenvironment{titledtboxh}[1]{\begin{tboxh}#1 \tboxhrule}{\end{tboxh}}

\makeatletter
\renewcommand{\paragraph}{%
  \@startsection{paragraph}{4}%
  {\z@}{10pt}{-1em}%
  {\normalfont\normalsize\bfseries}%
}
\makeatother

\makeatletter
\patchcmd{\@algocf@start}{-1.5em}{0pt}{}{}
\makeatother

\title{Half-Approximating Maximum Dicut in the Streaming Setting\footnote{All authors were supported in part by NSF CAREER Award CCF-2442812 and a Google
Research Award.}}

\author{
Amir Azarmehr\thanks{Northeastern University. Emails: \texttt{\{azarmehr.a, s.behnezhad, ferrante.s, saneian.m\}@northeastern.edu}.} \and
Soheil Behnezhad\footnotemark[2]  \and
Shane Ferrante\footnotemark[2] \and
Mohammad Saneian\footnotemark[2] 
}

\begin{document}

\date{}

\maketitle

\thispagestyle{empty}

\begin{abstract}
{
\setlength{\parskip}{0.5em}
    We study streaming algorithms for the maximum directed cut problem. The edges of an $n$-vertex directed graph arrive one by one in an arbitrary order, and the goal is to estimate the value of the maximum directed cut using a single pass and small space. With $O(n)$ space, a $(1-\epsilon)$-approximation can be trivially obtained for any fixed $\epsilon > 0$ using additive cut sparsifiers. The question that has attracted significant attention in the literature is the best approximation achievable by algorithms that use truly sublinear (i.e., $n^{1-\Omega(1)}$) space.
    
    A lower bound of Kapralov and Krachun (STOC'19) implies .5-approximation is the best one can hope for. The current best algorithm for general graphs obtains a .485-approximation due to the work of Saxena, Singer, Sudan, and Velusamy (FOCS'23). The same authors later obtained a $(1/2-\epsilon)$-approximation, assuming that the graph is constant-degree (SODA'25).

    In this paper, we show that for any $\epsilon > 0$, a $(1/2-\epsilon)$-approximation of maximum dicut value can be obtained with $n^{1-\Omega_\epsilon(1)}$ space in \textbf{general graphs}. This shows that the lower bound of Kapralov and Krachun is generally tight, settling the approximation complexity of this fundamental problem. The key to our result is a careful analysis of how correlation propagates among high- and low-degree vertices, when simulating a suitable local algorithm.
}
\end{abstract}

\vspace{5cm}

\paragraph{Independent work:} An independent and concurrent work of Velusamy \cite{velusamy2025near} gives a $(1/2-\epsilon)$-approximation of max dicut in $n^{1-\Omega_\epsilon(1)}$ space and two passes. Our algorithm has the same approximation/space trade-off  but runs in a single pass instead of two.

{
\newpage
\tableofcontents
\thispagestyle{empty}
}

\clearpage
\setcounter{page}{1}

\newcommand{\maxdicut}[0]{\ensuremath{\mathsf{MaxDiCut}}}
\newcommand{\CSP}[0]{\ensuremath{\mathsf{CSP}}}

\section{Introduction}

We study the {\em maximum directed cut} (\maxdicut{}) problem in the {\em streaming} setting.  This problem is a natural generalization of maximum cut to directed graphs. Specifically, the \maxdicut{} problem asks for a partition of vertices into two sets that maximizes the number of (directed) edges going from the first set to the second. In the streaming setting, the edges of this graph arrive one by one in an arbitrary order. The goal is to approximate the value of \maxdicut{} after taking a single pass over the input, while using a small space.

The \maxdicut{} problem has received significant attention from the streaming community over the last decade (see  \cite{guruswami2017streaming,ChouGV20,saxena-SODA23,saxena-FOCS23,saxena-SODA25,Madhu-Survey} and the references therein). Besides being a natural graph problem in its own right, \maxdicut{} serves as an important example of constraint satisfaction problems (\CSP{}) that have been studied extensively in the streaming setting \cite{kogan2014sketchingcutsgraphshypergraphs, KKS15, GVV17, KapralovKSV17, KapralovK19,bhaskara_et_al:LIPIcs.ICALP.2018.16, AKSY20, AN21, CGS+22b,ChouGSSV22, BHP+22, CKP+23, ChouGSSV22, KP22, saxena-SODA23, saxena-FOCS23, KPV23, Sin23,KPSY23,SSV24}. In particular, almost all general \CSP{} streaming algorithms have been extensions of algorithms that were first developed for \maxdicut{} \cite{saxena-SODA25,Madhu-Survey}. 

For the single-pass streaming \maxdicut{} problem, there is a folklore $O(n)$ space algorithm that achieves a $(1-\epsilon)$-approximation for any fixed $\epsilon > 0$.\footnote{The algorithm is to simply store $O(n/\epsilon^2)$ edges uniformly, which preserves all cuts to within an additive error of $\epsilon n$ with high probability, and then enumerate over them (in exponential time) to find the largest one. Note that this provides a multiplicative $(1-\epsilon)$-approximation of \maxdicut{}, since the \maxdicut{} value is large.} It is, therefore, natural to focus on algorithms with sublinear space and ask:
\begin{quote}
    {\em What is the best approximation of \maxdicut{} achievable in $n^{1-\Omega(1)}$ space?}
\end{quote}

On the lower bound side, a result of \citet{KapralovK19} implies that 1/2-approximation is the best one can hope for with sublinear space, already implying a separation. Despite numerous attempts and major progress \cite{guruswami2017streaming,ChouGV20,saxena-SODA23,saxena-FOCS23,saxena-SODA25}, existing algorithms do not match this lower bound. In particular, after a series of improvements, the current best known bound for general graphs is a .485-approximation due to the works of \citet*{saxena-SODA23,saxena-FOCS23}. 

Our main result in this work is to show that an (almost) 1/2-approximation can indeed be obtained for general graphs with truly sublinear space. This result, stated formally below as \cref{thm:main}, matches the lower bound of \cite{KapralovK19} and resolves the problem highlighted above regarding the approximability of \maxdicut{} with sublinear space.

\begin{graytbox}
\begin{theorem}\label{thm:main}
For any $\epsilon > 0$, there is a randomized one-pass streaming algorithm that, with high probability, outputs a $(\tfrac{1}{2}-\epsilon)$-approximation of \maxdicut{} using $n^{1-\Omega_\varepsilon(1)}$ space.
\end{theorem}
\end{graytbox}

We note that, prior to our work, \cref{thm:main} was proved by \cite{saxena-SODA25} under the assumption that the underlying graph $G$ has constant maximum degree. This assumption is crucial for the algorithm and analysis of \cite{saxena-SODA25}. Our main contribution is to completely remove this assumption, and achieve an (almost) 1/2-approximation for general graphs. We provide an overview of our approach and the challenges that arise along the way in \cref{sec:techniques}.

\subsection{Further Related Work}

Let us first review the literature on single-pass adversarial order streaming algorithms more extensively.  The first paper to achieve a non-trivial approximation for \maxdicut{} was the paper of \citet*{guruswami2017streaming} which obtained an (almost) $2/5$ approximation using $O(\log n)$ space. This was subsequently improved to $4/9$ by \citet*{ChouGV20} still using $O(\log n)$ space, who also showed obtaining a better approximation requires $\Omega(\sqrt{n})$ space. Finally, the approximation ratio was improved to .485 in the works of \citet*{saxena-SODA23,saxena-FOCS23} using $O(\sqrt{n})$ space.

We note that there are also a number of other results on some relaxations of the streaming setting for the \maxdicut{} problem. This includes algorithms that allow multiple passes over the input, or algorithms that work under the assumption that the input edges are ordered uniformly at random. We refer interested readers to the paper of \cite{saxena-SODA25} for an overview of these results.

\section{Technical Overview}\label{sec:techniques}

\subsection{Background}
A simple approach to solving the Max-DiCut problem would be to use the immediate neighborhood of a vertex (i.e., the ratio between the outgoing and incoming degree) to decide which side of the cut it appears on.
\citet{FeigeJ15} present such an algorithm that achieves a $0.483$-approximation, while showing that any approach based on the bias of the vertices (i.e., the ratio between the in-degree and the out-degree) cannot go beyond a $0.489$ ratio.
This suggests that to obtain a $(\frac{1}{2} - \epsilon)$-approximation, it is necessary to examine a larger radius around the vertices.
Below, we overview a line of work culminating in the $(\frac{1}{2} - \epsilon)$-approximation of \cite{saxena-SODA25} for constant-degree graphs in $n^{1-\Omega_\epsilon(1)}$ space.

Consider the following approach in the sequential (classic) setting.
The algorithm iterates over the vertices one by one.
Each vertex is assigned a position in the cut based on two things: (1) its \emph{degree information}, i.e., the number of incoming or outgoing edges to the assigned and unassigned vertices, and (2) the average position assigned to its neighbors that have already been processed.
\citet{BuchbinderFNS15} employ this method to obtain a $\frac{1}{2}$-approximation of Max-DiCut in linear time (in fact, they do so more generally for unconstrained submodular maximization).

\citet*{Censor-HillelLS17} derive a $(\frac{1}{2} - \epsilon)$-approximate LOCAL distributed algorithm from the sequential approach by utilizing vertex colorings.
The vertices are randomly colored using $k = O(\frac{1}{\epsilon})$ colors.
This can be turned into a valid coloring (i.e., such that adjacent vertices have different colors) by deleting any violating edges.
As a result, only $\epsilon m$ edges are removed in expectation, and hence Max-DiCut is largely unaffected, since it is at least $m/4$.
With the coloring at hand, the sequential algorithm can be implemented by ordering the vertices based on their color,
where the order between vertices of the same color is inconsequential, as none of them are adjacent.
Therefore, the position of each vertex can be computed recursively based on its information degree and the average position of its lower-color neighbors.
This yields a LOCAL distributed algorithm that computes the position of each vertex by examining its $k$-neighborhood, i.e., in $k$ rounds of communication.

We note the use of \emph{fractional positions} for these approximations.
That is, rather than assigning a 0-1 indicator variable to each vertex $u$ showing whether $u$ is on the source side of the cut, the algorithm assigns a fractional value $x_u \in [0, 1]$.
This could be interpreted as the probability of $u$ appearing on the source side of the cut.
As a result, the expected number of directed edges going out of the source side is equal to
$$
\sum_{(u, v) \in E} x_u (1-x_v).
$$

\citet{saxena-SODA25} simulate the $k$-round LOCAL algorithm of \cite{Censor-HillelLS17} to obtain a $(\frac{1}{2} - \epsilon)$-approximation for constant-degree graphs, using a space of $n^{1-\Omega_\epsilon(1)}$.
We assume a $k$-coloring of the graph, where $k = O(\frac{1}{\epsilon})$, and a function corresponding to the LOCAL algorithm, which outputs the position of a vertex based on its degree information (i.e., incoming and outgoing degree to lower-color and higher-color neighbors) and the average position of the lower-color neighbors.
For a vertex $v$, let $\pos(v)$ denote its position according to the LOCAL algorithm.
To compute $\pos(v)$, the streaming algorithm simulates the LOCAL algorithm by collecting the recursion tree that determines the position of a vertex $v$.
For a color-$1$ vertex, the recursion tree only consists of the vertex itself, and for higher colors, the recursion tree consists of the vertex plus the union of the trees of its lower-color neighbors.

Their algorithm relies on the small size of the recursion tree to collect it using random sampling.
Consider a graph with a constant maximum degree $\Delta$.
For any vertex, the size of the recursion tree is at most $\Tmax = O(\Delta^k)$ since the depth is at most equal to the number of colors $k$.
Therefore, if the vertices are sampled with probability $n^{-c}$ and the edges of the sampled vertices are stored completely,
then for any vertex, the entire recursion tree is collected with probability at least $n^{-c\Tmax}$.
If the recursion tree of a vertex is collected, then the LOCAL algorithm can be exactly simulated to compute the position of the vertex.
In this case, we say the vertex has succeeded.
When $c$ is chosen appropriately, for a sufficiently large number of edges, both endpoints succeed, and the cut value can be computed.
This yields a single-pass algorithm with memory $O(n^{1 - c})$.

\subsection{Our Techniques}

In addition to the $k$-coloring of the graph and the function corresponding to the LOCAL algorithm, we assume that the graph has $O(n/\epsilon^2)$ edges. This is without loss of generality, since we can sample the edges with probability $O(\frac{n}{\epsilon^2 m})$ which preserves the value of Max-DiCut up to a factor of $(1 + O(\epsilon))$ with high probability (where the value of $m$ can be guessed up to a constant by running parallel copies of the algorithm, so we can assume we know $m$ w.l.o.g.).

On a high level, our algorithm sparsifies the recursion trees and estimates the position of a vertex rather than exactly simulating the LOCAL algorithm.
A great deal of care is required to manage the correlation between the estimates.\footnote{
We remark that \cite{saxena-SODA25} also presents a multi-pass algorithm that sparsifies the recursion trees. However, their approach does not extend to the single-pass setting, as it relies heavily on collecting the layers of the recursion trees pass by pass in a BFS manner.}
We divide the vertices into two groups based on their degree, and handle each group differently.
A vertex of color $a$ is \emph{mildly low-degree} if it has a degree of at most $n^{q2^a}$ for a fixed constant $q$ dependent on $\epsilon$.\footnote{More precisely, the edges are sampled with probability $n^{-c}$, and the vertex is considered low-degree if it has at most $n^{q2^a}$ sampled edges. We ignore this detail here for the sake of simplicity.}
Otherwise, it is \emph{high-degree}.
Note that our notion of low-degree here includes vertices of degrees polynomial in $n$, i.e.\ it is not restricted to constant-degree vertices and still requires sparsification.

\paragraph{Mildly low-degree vertices. }
First, we discuss low-degree vertices.
Every low-degree vertex is \emph{sampled} with probability $n^{-c}$ to form the vertex set $W$.
The degree information is maintained exactly for the vertices in $W$.
Then, for any low-degree sampled vertex $v \in W$, a set of $d = O_\epsilon(1)$ lower-color neighbors are sampled with replacement, referred to as the \emph{selected neighbors}.

Consider, for now, a case where all vertices are low-degree.
With $W$ and the selected neighbors at hand, the position of a vertex $v \in W$ is estimated as follows.
Recall that the exact position $\pos(v)$ is determined by the degree information, which is computed exactly for all vertices in $W$, and the average position of its lower-color neighbors.
The latter is estimated through the selected neighbors of $v$.
Note that, as a result, estimating the position of $v$ would require first estimating the position of its selected neighbors.
This dependency motivates the notion of \emph{succeeding} for a vertex:
we say a vertex succeeds if it is sampled, and all its selected neighbors also succeed.
Put differently, the selected neighbors define a sparsified recursion tree, and a vertex succeeds if all the vertices in its sparsified tree are sampled in $W$, in which case the estimated position can be computed.

We give a high-level analysis of this special case, where all the vertices are low-degree.
For the sake of analysis, one can imagine that every vertex (not just those in $W$) selects $d$ lower-color neighbors.
Note that, similarly to the constant-degree case, the sparsified recursion trees have a size of at most $\Tmax = O(d^k)$.
Therefore, when $c$ is chosen appropriately, the sparsified tree of each vertex is collected in $W$ with a sufficiently large probability.
We note that we are estimating the position of each vertex at most once, and passing up the same estimate for each recursive call.
That is, if two vertices $u$ and $v$ select the same lower-color neighbor $w$, the same estimate of the position of $w$ is passed up to $u$ and $v$.
As a result, our estimates may be correlated, which could potentially make the averages unreliable (i.e., high-variance).
However, we manage to bound the variance inductively by allowing a variance of $\sigma^2_a$ for a vertex of color $a$, which is larger for higher colors.
This issue of correlated averages is exacerbated with the high-degree vertices as we discuss below.

\paragraph{High-degree vertices. }
Next, we move on to high-degree vertices.
As opposed to low-degree vertices, which might fail to produce an estimate with some probability, high-degree vertices always succeed.
Recall that to estimate the position of a vertex $u$, we require the degree information and the average position of the lower-color neighbors.
By sampling all the edges with probability $n^{-c}$, we can accurately estimate the degree information of all high-degree vertices.
To estimate the average position of the lower-color neighbors, first, we use the same edge samples and obtain a sample of lower-color neighbors $\Neighbors(u)$.
For each selected neighbor $v \in \Neighbors(u)$, let $q_v$ be the probability that $v$ succeeds,\footnote{In the algorithm, the probability of success is exactly determined by the color of the vertex. We defer how this is achieved to the main body of the paper.} 
and let $\Pest(v)$ denote its estimated position when it does.
To estimate the average of the neighbors, we sum $\Pest(v) / q_v$ for the neighbors $v \in \Neighbors(u)$ that succeed, and divide by $\card{\Neighbors(u)}$.
The estimates are scaled by $1/q_v$ to compensate for the cases where the vertex fails, so that each vertex contributes about $\pos(v)$ to the sum in expectation.

We reiterate that for high degree vertices $v$, the success probability $q_v$ is equal to $1$.
As a result, in the simple case where all the vertices are high-degree, the method above trivially takes the average estimate for all $v \in \Neighbors(u)$.
In this case, the high degree of $u$ ensures that enough neighbors are sampled for this estimate to be accurate.
However, for general values of $q_v$, i.e.\ when some of the lower-color neighbors are low-degree and $q_v < 1$, managing the variance of our estimate becomes more challenging.
That is, when the estimates are scaled by $1/q_v$,
the variance of the estimate for $v$ grows by $1/q_v^2$,
and the covariance of the estimates for $v_1$ and $v_2$ grows by $1/q_{v_1}q_{v_2}$.
This may increase the variance of our average significantly, as $1/q_v$ could be as large as $n^{c\Tmax}$, where $\Tmax = O(d^k)$.
To illustrate, in part, how the variances are analyzed with the rescalings, we present a special case below.

Consider a special case where for a high-degree vertex $u$, the entire (non-sampled) recursion tree consists of low-degree vertices.
Let $v_1, \ldots, v_{\card{\Neighbors(u)}}$ denote the selected neighbors of $u$.
If two of these neighbors $v_i$ and $v_j$ include the same vertex $w$ in their sampled recursion tree.
Our algorithm estimates the position of $w$ only once, and the same estimate is used for $v_i$ and $v_j$.
As a result, the estimates for $v_i$ and $v_j$ may be correlated.
In contrast with the case where all the vertices were low-degree, this correlation is substantial, as it will be scaled by $1/{q_{v_i}}q_{v_j}$.

To address this challenge, the number of pairs $(v_i, v_j)$ with intersecting subtrees can be bounded as follows.
Each vertex $v_i$ has a subtree of size at most $\Tmax = O(d^k)$, due to the fact that all vertices in the recursion tree of $u$ are low-degree.
Moreover, each vertex in the subtree of $v_i$ can be in the subtree of at most $n^{q2^{a-1}}$ other vertices $v_j$, where $a$ denotes the color of $u$
(this relies crucially on the threshold for low-degree vertices, which grows doubly-exponentially in $a$).
Therefore, the number of intersecting pairs is at most $O(d^k n^{q2^{a-1}})$, which is small enough since $u$ is taking average over approximately $n^{q2^a}$ neighbors.

The main challenge of our analysis revolves around the full interaction between high-degree and low-degree vertices, i.e., when high-degree vertices appear in the sparsified recursion tree.
For a vertex $u$, many of its lower-color neighbors may be correlated as the sparsified recursion tree is no longer entirely constant-degree (in fact, every pair of neighbors may be correlated as a high-degree vertex of degree $n-1$ can connect everything).
However, we prove that any such correlation that \enquote{goes through a high-degree vertex} is small in magnitude.
That is, while many of the lower-color neighbors may be correlated, the sum of the correlations (and hence the variance of the estimated average), is still small.

The analysis proceeds in two stages.
In the first stage, we reveal the randomness of selecting the neighbors (\cref{sec:p1-p2}).
In this stage, based on the selected neighbors (which determines the sparsified trees), we define a set of intermediary estimates that correspond to a hypothetical scenario, where all the vertices succeed, i.e., every vertex is in $W$ and there is no sampling of vertices.
We prove that these estimates are close (with constant probability) to the exact positions computed by the LOCAL algorithm.
In the second stage, we reveal the randomness of sampling the vertices, i.e., the randomness of $W$ (\cref{sec:p2-p3}).
Combined with the sparsified trees, this determines the final estimates that the algorithm computes.
We show that these estimates are close (with constant probability) to the intermediary estimates.

In each stage, we analyze the expectation and the variance of our estimates using an induction on the color.
With larger colors, the estimates are allowed to be more biased (i.e., the expected value can be further from the exact position), and less accurate (i.e., the variance can be larger).
In the second stage, to bound the correlation between different estimates, we prove that the estimate for a vertex $v$ is not largely affected by conditioning on the success of other vertices.
More precisely, we show that the shift in the conditional expectation and the growth of the conditional variance are proportional to the number of vertices we condition on (linear for the expected value, and quadratic for the variance). 
\section{Preliminaries}

\paragraph{The problem.} Given a directed graph $G=(V, E)$, the \maxdicut{} is a vertex subset $S \subseteq V$ maximizing the number of (directed) edges that go from $S$ to $V \setminus S$. Let $\opt(G)$ denote the number of such edges in \maxdicut{}. For $\alpha \in (0, 1]$, we say an algorithm provides an $\alpha$-approximation for \maxdicut{} if given any graph $G$ it provides an estimate $s$ such that $\alpha \cdot \opt(G) \leq s \leq \opt(G)$.

\paragraph{Streaming setting.}
Let the vertices have labels $1, 2, \ldots, n$.\footnote{While this assumption is standard, it is not crucial for our algorithm. Assuming the vertex labels are unique elements from a larger set of size $N = \poly(n)$, our algorithm needs the following modification. To sample the vertices with probability $n^{-c}$, rather than directly sampling from $\{1, 2, \ldots, n\},$ the algorithm can use $n^{1-c}$-independent hash functions.}
The graph is given to the algorithm as a stream of edges, each denoted by an ordered pair of vertex labels.
The algorithm reads the edges from the stream one by one and outputs an approximation of \maxdicut{} at the end.
The goal is to produce a good approximation while using a small amount of space.
Our algorithm produces a $(\frac{1}{2} - \epsilon)$-approximation using $n^{1 - \Omega_\epsilon(1)}$ space.

It can be assumed, by incurring an $O(\log n)$ factor in the space complexity, that the value of $m$ (the number of edges) is known up to a constant.
To do so, we can run $O(\log n)$ independent copies of the algorithm, where the $i$-th copy assumes $m \in [2^{i-1}, 2^i[$.
If the copy misbehaves (i.e.\ tries to use more than $n^{1 - \Omega_\epsilon(1)}$ space) or $2^i$ edges arrive in the stream, then it is terminated.
At the end of the stream, the value of $m$ is known, and we can use the output from the appropriate copy.
For simplicity, the rest of the paper assumes that $m$ is known exactly. This affects only the constants in our algorithm and analysis.
We also disregard the logarithmic factor in the space, as the significance of our result is that the space complexity is \emph{strongly sublinear} in $n$.

\begin{definition}
    Given a vertex $v$ in graph $G=(V,E)$, we define several notions of degree. Let $d(v)$ be the degree of $v$. Let $\indeg(v)$ and $\outdeg(v)$ be the incoming and outgoing degrees of $v$, respectively. For some set of edges $F \subseteq E$, we define $d_F(v), d_F^-(v),$ and $d_F^+(v)$ similarly, only counting the edges in $F$. Lastly, if $G$ has a vertex coloring $\chi:V\to[k]$ then for color $a\neq \chi(v)$, define $d_a(v), d_a^-(v),d_a^+(v)$ similarly, only counting edges adjacent to vertices of color $a$.
\end{definition}

\begin{definition}
    For a graph $G=(V,E)$, and any vertex function $f : V \to [0,1]$, we define a corresponding function on the edges with the same name $f: E \to [0,1]$ where, for edge $e=(u,v)$, 
    \[
    f(e) = f(u)\cdot(1-f(v)).
    \]
\end{definition}

\begin{definition}
    For a graph $G=(V,E)$, and any vertex function $f : V \to [0,1]$, define $\val{G}{f}$ as follows:
    \[
    \val{G}{f} = \frac{1}{|E|}\sum_{e \in E} f(e) = \frac{1}{|E|}\sum_{(u,v) \in E} f(u) \cdot (1-f(v)).
    \]
\end{definition}

\begin{definition}
    For a graph $G=(V,E)$, $\maxval{G}$ is the maximum over all boolean functions $f:V \to \{0,1\}$ of $\val{G}{f}$.
\end{definition}

We see that this is equal to $\maxdicut$ for graph $G$, when scaled by the number of edges in the graph.

\begin{proposition}
    For a graph $G=(V,E)$, let $f$ any vertex function $f : V \to [0,1]$, there exists boolean functions $g,g': V \to \{0,1\}$ with 
    \[
    \val{G}{g} \leq \val{G}{f} \leq \val{G}{g'}.
    \]
\end{proposition}

This implies that $\maxval{G}$ is equivalently the maximum over all fractional assignments $f:V \to [0,1]$. We see that $\maxval{G} \leq 1$. We can also trivially lower-bound $\maxval{G}$ by $1/4$ using the fractional cut $f(v) = 0.5$ for all vertices, giving $\val{G}{f} = 1/4$. Thus, our \maxdicut{} value for any graph is at least $m/4$.

\begin{definition} Given $x \in \mathbb{R}$, we define
\[
\clamp{x} := 
    \begin{cases}
        0 & \text{if } x \leq 0, \\ 
        x & \text{if } 0 < x < 1, \\
        1 & \text{if } x \geq 1.
    \end{cases}
\]
    
\end{definition}


\section{Basic Setup and Initialization}

\subsection{Constants}\label{sec:constants}

There are a number of constants that are necessary parameters of the algorithms used in this paper. These constants are defined in \Cref{table:values}. 

\setlength{\arrayrulewidth}{1pt}
\setlength{\tabcolsep}{12pt}
\renewcommand{\arraystretch}{1.5}

\begin{table}[h!]
\centering
\begin{tabularx}{\textwidth}{
    |>{\columncolor{gray!50}}c
    |>{\columncolor{gray!50}}c
    |*{7}{>{\centering\arraybackslash}X|}
}
\hline
\rowcolor{white}
\textbf{Constant} &
\textbf{$k$} &
\textbf{$\alpha$} &
\textbf{$\delta_a$} &
\textbf{$\sigma_a$} &
\textbf{$d$} &
\textbf{$\T_a$} &
\textbf{$q$} &
\textbf{$c$} \\ 
\hline

\rowcolor{white}
\textbf{Value} &
$1/\epsilon^2$ &
$\epsilon^4$ &
$\epsilon^{4^{k+2-a}}$ &
$\delta_a^2$ &
$\epsilon^{-4^{k+3}}$ &
$(2d)^a$ &
$2^{-(k+1)}$ &
$q/10\T_k$ \\
\hline
\end{tabularx}
\caption{Value of parameters used throughout the paper.}
\label{table:values}
\end{table}

We discuss some of the key relationships between these constants and their main functions in relative increasing scale.

\begin{itemize}
\item $\epsilon$ is our main parameter, our algorithm produces a $(1/2-\epsilon)$-approximation. 

\item $k$ is the number of colors in our colored graph (formalized in \cref{sec:simplifying-coloring}).

\item $\alpha$ is a parameter defined for \Cref{alg:pos-def}, this is a deterministic non-streaming algorithm that our algorithm simulates, which gives a $(1/2-\alpha)$-approximation.

\item $\delta_a,\sigma_a$ are families of constants used for inductively bounding the mean and variance of values in the analysis. In general, $\delta_a, \sigma_a$ are used for vertices of color $a$, however, we also refer to the constants $\delta_0,\sigma_0,\delta_{k+1},\sigma_{k+1}$ despite all vertices having color from 1 to $k$.

\item $d$ is the constant that determines the number of sampled neighbors we take for low-degree vertices.

\item $\T_a$ is a constant that determines the number of vertices a low-degree vertex of color $a$ relies on being sampled for an accurate estimate. We also refer to this as the size of the \enquote{low-degree-tree} of our vertex. This also determines the probability that we are able to estimate the position of a vertex of color $a$, as each vertex is sampled independently with probability $n^{-c}$, this probability is $n^{-c\T_a}$. We also refer to the constant $\Tmax = \T_k$. 

\item $q$ is a constant used for the threshold that determines if a vertex is high-degree. A vertex of color $a$ is high-degree (loosely speaking) if its degree is greater than $n^{q\cdot 2^a}$. 

\item $c$ is our smallest constant which determines sampling probabilities and hence our space bound. Our algorithm has a space complexity of $O(n^{1-c})$, where $c = \exp(\exp(1/\epsilon))^{-1}$. 

\item Lastly, for convenience, we assume that $\epsilon \leq 0.01$ and $n$ is large enough such that $\exp\paren*{-cn^c} \leq c$ so that all Chernoff bounds give small enough probabilities.
\end{itemize}



\subsection{Simplifying and Coloring the Input Graph}\label{sec:simplifying-coloring}

Let $G=(V, E)$ be an $n$-vertex $m$-edge graph. In this section, we prove that we can instead estimate the max directed cut value of a vertex-colored graph $G'=(V,E')$ with the following three properties without loss of generality.

\begin{enumerate}[label=$(A\arabic*)$]
    \item Graph $G'$ has at most $10(n/\epsilon^4)$ edges and at least $n^{1-c}$. \label{assumption:linearedges}
    \item For every vertex $v \in V$ with $d(v) \geq n^{2q}$, we have $d^-(v), d^+(v) \geq \epsilon^2d(v)$.\label{assumption:balanced-in-out}
    \item For every vertex $v \in V$ with $d(v) \geq n^{2q}$, the colors of the neighbors are evenly distributed. That is, for a color $a \neq \chi(v)$, we have 
    $$\abs*{d_a^-(v)-d^-(v)/(k-1)} \leq \delta_0d(v)/(k-1) \quad \text{and} \quad \abs*{d_a^+(v)-d^+(v)/(k-1)} \leq \delta_0d(v)/(k-1).$$ 
    \label{assumption:balanced-colors}
\end{enumerate}

More specifically, we prove the following:

\begin{lemma}\label{lem:streaming-reduction}
    Suppose there is a streaming algorithm that $\paren*{\frac{1}{2}-17\epsilon^2}$-approximates the value of maximum directed cut for a colored-graph under assumptions \ref{assumption:linearedges}, \ref{assumption:balanced-in-out}, and \ref{assumption:balanced-colors} using space $O(n^{1-c})$ and failure probability $p$. Then there is an algorithm that $\paren*{\frac{1}{2}-\epsilon}$-approximates the maximum directed cut value on general graphs, using $O(n^{1-c})$ space, and with failure probability $\epsilon$.
\end{lemma}

\begin{proof}
We reduce an arbitrary input graph $G=(V,E)$ to one satisfying assumptions 
\ref{assumption:linearedges}, \ref{assumption:balanced-in-out}, and \ref{assumption:balanced-colors}, while preserving the maximum directed cut up to a small multiplicative factor.

\paragraph{Step 1. Reducing the number of edges.}
If $G$ has fewer than $10n^{1-c}$ edges, we can simply store all of them and we can compute our maximum directed cut exactly, so we assume $G$ has at least $10n^{1-c}$ edges. If $G$ still has fewer than $2n/\epsilon^4$ edges, leave it as is. Otherwise, form a random subgraph $G_2$ by sampling each edge independently with probability 
\[
p = \frac{n}{\epsilon^4 m}.
\]
Let $\mathrm{OPT}_2$ be the maximum directed cut in $G_2$. Fix a cut $S \subseteq V$ with $\delta$ crossing edges. Define independent Bernoulli random variables $X_i$ indicating whether edge $i$ is sampled and crosses $S$, and set
$X = \sum_{i=1}^m X_i$ with expectation $\mathbb{E}[X]=p\delta$. By a Chernoff bound (\Cref{lem:add-chernoff}), 
\[
\Pr\paren*{|X - p\delta| \ge 3p \epsilon^2 m} \leq 2 \exp \paren*{-\frac{(3p\epsilon^2 m)^2}{3p\delta}} \leq 2\exp \paren*{-3n}.
\]
Applying a union bound over all $2^n$ cuts, we obtain
\begin{equation}\label{eq:On-edge-reduction-union-bound}
    \Pr\paren*{\exists\ \text{cut $S$ with error $\geq$ $3p\epsilon^2 m$}} \le 2\exp(-n) \leq \epsilon^2.
\end{equation}

Thus with high probability, if we scale by $1/p$, every cut value is preserved up to additive $\pm 3 \epsilon^2 m$, which implies $p\mathrm{OPT}_2 \in (1 \pm 12\epsilon^2)\mathrm{OPT}$ as $\mathrm{OPT} \geq m/4$. In addition let $m_2$ be the number of edges in $G_2$, by \Cref{lem:mult-chernoff}, we know that $\Pr(\abs*{m_2-pm} > pm/2) \leq 2\exp \paren*{-\frac{pm}{12}} = 2\exp \paren*{-\frac{n}{12\epsilon^4}} \leq \epsilon^2$. Thus, if this holds, $10n^{1-c} \leq m_2 \leq 2n/\epsilon^4$. This means that we have reduced our graph to a graph $G_2$ with property \ref{assumption:linearedges} with multiplicative error $(1 \pm 12\epsilon^2)$ and failure probability $2\epsilon^2$.

\paragraph{Step 2. Balancing high-degree vertices.}
Next, independently flip the orientation of each edge of $G_2$ with probability $3\epsilon^2$, and let $G_3$ denote the resulting graph. For a vertex $v$. In $G_3$, let $X_v$ denote these incoming such neighbors. Then $\mathbb{E}[X_v]\ge 3\epsilon^2 d(v)$, but at most $d(v)$. By a Chernoff bound (\Cref{lem:add-chernoff}), we have
\[
\Pr \paren*{X_v < 2\epsilon^2d(v)} \leq 2\exp \paren*{-\frac{\epsilon^4d(v)^2}{3d(v)}}.
\]

Hence if $d(v)\ge n^{2q}$, then with high probability both the in-degrees and out-degrees of $v$ are at least $2\epsilon^2d(v)$. A union bound implies that all high-degree vertices have each direction balanced with failure probability at most $4nk\exp\paren*{-\frac{2\epsilon^4}{3}n^{2q}} \leq \epsilon^2$. Now, let $X$ be the total number of flipped edges, we see that $\Ex[X] = 3\epsilon^2m_2$, and by \Cref{lem:add-chernoff}, 
\[
\Pr \paren*{X > 4\epsilon^2 m_2} \leq 2\exp\paren*{-\frac{\epsilon^4m_2^2}{3\cdot3\epsilon^2m_2}} \leq \epsilon^2.
\]
This implies that with high probability, $\mathrm{OPT}_3 \in (1 \pm 16\epsilon^2)\mathrm{OPT}_2$ ($\mathrm{OPT}_2 \geq m_2/4$). Thus, we have reduced our graph $G_2$ to a graph $G_3$ with properties \ref{assumption:linearedges} and \ref{assumption:balanced-in-out} with multiplicative error $(1 \pm 16\epsilon^2)$ and failure probability $2\epsilon^2$.

\paragraph{Step 3. Coloring the Graph.} 

We use $2r$-independent uniform hash functions to obtain the coloring, where $r = O_{\epsilon}(1)$ is specified shortly.
That is, starting from $G_3$, create a colored graph $G'$ by coloring each vertex randomly with a color $\chi(v)$ from $1$ to $k$ and removing all edges $(u,v)$ where $\chi(u) = \chi(v)$, making it a proper coloring. 
We remark that storing the coloring takes $O(r \log n)$ space.
For a vertex $v$ with $d(v) \geq n^{2q}$, let $d^-(v)$ denote the number of incoming edges, and let $d_a^-(v)$ denote the number of incoming neighboring vertices of color $a$ (all degrees are considered before removing the violating edges). To obtain \ref{assumption:balanced-colors}, we use the Markov bound on the $2r$-th moment of $d_a^-(v)$, and then use the union bound over all vertices and colors to show the assumption holds.\footnote{One could more generally use Chernoff-like bounds for variables with limited dependence \cite{SchmidtSS95,Skorski22}.}

As each incoming neighbor takes color $a$ with probability $\frac{1}{k}$, it holds that $\Ex[d_a^-(v)] = \indeg(v)/k$.
Since the colors of the vertices are $2r$-independent, the $2r$-th central moment of $d_a^-(v)$ is the same as that of a binomial distribution with $d^-(v)$ trials and $1/k$ success probability.
That is:
$$
\Ex\left[(d^-_a(v) - d^-(v) / k)^{2r}\right] = (2r - 1)!! \sigma^{2r} + O((d^-(v))^{r-1})
= O_{r, \epsilon}(1) \cdot (d^-(v))^r,
$$
where $(2r-1)!! = \frac{(2r)!}{2^rr!}$ denotes the product of odd integers up to $2r - 1$, and $\sigma^2 = d^-(v) \cdot \frac{1}{k} \cdot(1 - \frac{1}{k})$.
Applying the Markov bound to the $2r$-th moment gives us the following bound:
\begin{align}
\Pr \paren*{\abs*{d_a^-(v)-d^-(v)/k} > \delta_0d^-(v)/k}
&\leq \frac{\Ex\left[(d^-_a(v) - d^-(v) / k)^{2r}\right]}{(\delta_0d^-(v)/k)^{2r}} \notag \\
&\leq O_{r, \epsilon}(1) \cdot \big(d^-(v)\big)^{-r} \notag \\
&\leq O_{r, \epsilon}(1) \cdot n^{-2qr} \label{eq:color-balance},
\end{align}
where the last inequality follows from \ref{assumption:balanced-in-out}, since $d(v) \geq n^{2q}$.
Therefore, applying the union bound to all vertices $v$ with $d(v) \geq n^{2q}$ and all colors, implies that $\abs*{d_a^-(v)-d^-(v)/k} \leq \delta_0d^-(v)/k$ holds for all of them with probability $O_{\epsilon, r}(1) \cdot  n^{1-2qr}$.
Hence, letting $r = 10/q = O_{\epsilon}(1)$, it holds with high probability.


To show \ref{assumption:balanced-colors} is satisfied, we must adapt this bound to one that uses our new degrees after edges are removed. We see that $d'^-(v) = d^-(v)-d_{\chi(v)}^-(v)$ which changes our bound by a constant factor. Thus, our new estimate is $d'^-(v)/(k-1)$, and we double the error on our new bound for simplicity, giving $\abs*{d'^-_a(v)-d'^-(v)/(k-1)} \leq \frac{2\delta_0 d'^-(v)}{k-1} \leq \frac{2\delta_0 \epsilon^2d'(v)}{k-1} \leq \frac{\delta_0d'(v)}{k-1}$. This implies that $G'$ has property \ref{assumption:balanced-colors}, as well as preserving properties \ref{assumption:linearedges} and \ref{assumption:balanced-in-out}. 

Finally, to show that the \maxdicut{} is preserved, we must prove that only a few edges are removed in this transformation (i.e.\ only a few edges violate the random coloring).
For an edge $e$, let $Y_e$ be the indicator variable that it is removed.
Since the coloring is uniform and 2-independent, it holds that $\Ex[Y_e] = 1/k$.
Let the total number of removed edges be $Y = \sum_{e\in E(G_3)} Y_e$.
Since the value of $Y_e$ is determined by the color of its endpoints, having $2r \geq 4$ implies that the indicator variables for the edges are also 2-independent.
Therefore, the variance of $Y$ can be bounded by $m_2 \frac{1}{k}(1-\frac{1}{k})$. Applying the Chebyshev bound, we get:
\[
\Pr \paren*{\abs*{Y-m_2/k} > m_2/k} \leq \frac{\Var(Y)}{(m_2/k)^2}
\leq \frac{k(1 - \frac{1}{k})}{n^{1 - c}} \leq \epsilon^2.
\]

Thus, let $\mathrm{OPT}_{G'}$ be the maximum directed cut in $G'$. If the bound above does not fail, then we know that $\mathrm{OPT}_{G'} \in \paren*{1 \pm \frac{8}{k}}\mathrm{OPT}_3 \subseteq \paren*{1 \pm 8\epsilon^2}\mathrm{OPT}_3$. Thus, we have reduced $G_3$ to $G'$ with properties \ref{assumption:linearedges}, \ref{assumption:balanced-in-out}, and \ref{assumption:balanced-colors}, multiplicative error $\paren*{1 \pm 8\epsilon^2}$ and failure probability $2\epsilon^2$.

To conclude the proof, we see that our total multiplicative error is $\paren*{1 \pm 12\epsilon^2}\cdot\paren*{1 \pm 8\epsilon^2}\cdot\paren*{1 \pm 16\epsilon^2} \subseteq \paren*{1\pm 37\epsilon^2}$. We now use our assumption that there is an algorithm that gives a $\paren*{1/2-17\epsilon^2}$ approximation of \maxdicut{} with failure probability $\epsilon^2$. Our reduction implies that we can return a value $\Out$ in the following range.
\[
(1-37\epsilon^2)\cdot\paren*{1/2-17\epsilon^2}\cdot\opt \leq \Out \leq (1+37\epsilon^2)\opt.
\]
Thus, we can divide our $\Out$ by $(1+37\epsilon^2)$ to get a value in the range
\[
\paren*{1/2-\epsilon}\cdot\opt \leq \Out \leq \opt,
\]
with probability $1-\epsilon$. This concludes the proof.
\end{proof}

Thus, for the remainder of the proof, we can assume that we are given a colored graph $G'$ with assumptions \ref{assumption:linearedges}, \ref{assumption:balanced-in-out}, and \ref{assumption:balanced-colors}.


\section{An Offline Algorithm for Estimating the Max Dicut Value}

We begin by recalling an offline algorithm of \cite[Algorithm 4]{saxena-SODA25} that is deterministic and in the general non-streaming setting to define the fractional cut position values ($\pos$). Though we would like to use this algorithm as a black box, we will make some modifications later, so an exact formulation of the algorithm and its prerequisites are required here. Despite this, we leave the more detailed description to \cite{saxena-SODA25} as we claim no novelty regarding this formulation.

First, we define a partition of the edges adjacent to a vertex $v \in V$ as follows:
\[
\begin{aligned}
\IL(v) &\coloneqq \set*{ (u, v) \in E \mid \chi(u) < \chi(v) },
&\hspace{1cm} \OL(v) &\coloneqq \set*{ (v, u) \in E \mid \chi(u) < \chi(v) } , \\
\IH(v) &\coloneqq \set*{ (u, v) \in E \mid \chi(u) > \chi(v) }, &\hspace{1cm} \OH(v) &\coloneqq \set*{ (v, u) \in E \mid \chi(u) > \chi(v) }. 
\end{aligned}
\]

Next, for all $v \in V$, we define:
\begin{equation}\label{eq:yinout-def}
\begin{aligned}
\yin(v) &\coloneqq \max\paren*{ \abs*{ \IH(v) }, \alpha \cdot \abs*{ \IL(v) } }, \\
\yout(v) &\coloneqq \max\paren*{ \abs*{ \OH(v) }, \alpha \cdot \abs*{ \OL(v) }}.
\end{aligned}
\end{equation}

The intuition behind these definitions is, in almost all cases $\yin(v) = \abs*{\IH(v)}$, and $\yout(v) = \abs*{\OH(v)}$. The reason we add this $\alpha$ term is so that we can get a lower-bound for $\yin(v) + \yout(v)$ in terms of $\abs*{\IL(v)}+\abs*{\OL(v)}$, which we have in \Cref{eq:yinout-simple}. This will help us in the analysis because, as we will see in our algorithm, the size of $\zli(v)$ and $\zlo(v)$ scale corresponding to $\abs*{\IL(v)}+\abs*{\OL(v)}$.

We may assume that $\yin(v) + \yout(v) > 0$, since we never need to compute the position of isolated vertices (and their position could be set arbitrarily).
Thus, we know the following basic identities about $\yin(v)$ and $\yout(v)$
\begin{equation}\label{eq:yinout-simple}
\begin{aligned}
    0 < \yin(v) + \yout(v) &\leq d(v), \\
    \alpha\cdot\paren*{\abs*{\IL(v)}+\abs*{\OL(v)}} &\leq \yin(v) + \yout(v).
\end{aligned}
\end{equation}

Now we can describe the algorithm. For a directed graph $G'=(V,E)$ with proper k-coloring $\chi:V \to [k]$ of $G$, the function $\pos : V \to [0,1]$ is defined recursively according to the algorithm described below.

\begin{algorithm}[H]
    \caption{A recursive, deterministic procedure to define the fractional cut $\pos$}
    \label{alg:pos-def}
    \KwIn{Graph $G'=(V,E)$, a proper k-coloring $\chi:V \to [k]$ of $G$, and a vertex $v \in V$}
    \KwOut{The assignment $\pos(v) \in [0,1]$}
    
    \DontPrintSemicolon
    \SetAlgoSkip{bigskip}
    \SetAlgoInsideSkip{}

    Compute Recursively:

    \[
    \zli (v) = \displaystyle \sum_{(u,v) \in \IL(v)}
    \pos(u), 
    \hspace{1cm}
    \zlo (v) = \displaystyle \sum_{(v,u) \in \OL(v)}
    (1-\pos(u)).
    \]

    Output:
    $$
    \pos(v) = \clamp{\frac{\yout(v) + \zlo(v) - \zli(v)}{\yin(v)+\yout(v)}}
    $$

\end{algorithm}

Now that the algorithm is defined, we can use the following lemma from \cite{saxena-SODA25}\footnote{Lemma 7.9 in the paper.} of its correctness without proof.

\begin{lemma}\label{lem:half-apx-valGpos}
Let $k > 0$, $G = \paren*{ V, E }$ be a graph, and $\chi : V \to [k]$ be a proper coloring of $G$. Let $\alpha \geq 0$ be arbitrary and $\pos : V \to [0,1]$ be as defined by \cref{alg:pos-def}. We have:
\[
\paren*{ \frac{1}{2} - \alpha } \cdot \maxval{G} \leq \val{G}{\pos} \leq \maxval{G} .
\]
\end{lemma}

\subsection{Modification for high-degree vertices}

We will now give an estimate to this formula for high degree vertices which makes use of the assumption that our coloring of the vertices is random. 

First, we define $\zb(v)$ as the average position value of the lower-colored neighbors (if there are no lower-colored neighbors, we say $\zb(v) = 0$).
\begin{equation}\label{eq:zb-def}
    \zb(v) = \frac{1}{\abs*{\IL(v)} + \abs*{\OL(v)}} \left( \displaystyle \sum_{(u,v) \in \IL(v)} \pos(u) + \displaystyle \sum_{(v,u) \in \OL(v)} \pos(u) \right).
\end{equation}

Next, we can define the estimate $\tpos(\cdot)$, we note that this is still an offline estimate and we are not yet in the streaming setting.

\begin{equation}\label{eq:tpos-def}
    \tpos(v) =
    \altclamp
    \left(
    \begin{cases}
            \frac{\alpha+1}{\alpha} \cdot\frac{\outdeg(v)}{d(v)} - \frac{\zb(v)}{\alpha} & \text{if } \chi(v) = k \\ 
            \frac{k-1}{k-\chi(v)} \cdot \frac{\outdeg(v)}{d(v)} - \frac{\chi(v)-1}{k-\chi(v)} \cdot \zb(v) & \text{if } \chi(v) \neq k
    \end{cases}
    \right).
\end{equation}

Now, we prove that the estimate is accurate for high-degree vertices.

\begin{lemma}\label{lem:pos-tpos}
    For a vertex $v$ with $d(v) \geq n^{2q}$,
    \[
    \abs*{\pos(v)-\tpos(v)} \leq 3\delta_0 k^2.
    \]
\end{lemma}

\begin{proof}
    The first case just takes some manipulation of the formula for $\pos(v)$. We begin with the term $\zlo(v)-\zli(v)$, which we can rearrange and substitute using \cref{eq:zb-def}.
    \begin{equation}\label{eq:zlozli-rearranged}
    \begin{aligned}
        \zlo(v) - \zli(v) &= \displaystyle \sum_{(v,u) \in \OL(v)} (1-\pos (u)) - \displaystyle \sum_{(u,v) \in \IL(v)} \pos (u) \\
        &= |\OL(v)| - \left( \displaystyle \sum_{(v,u) \in \OL(v)} \pos (u) + \displaystyle \sum_{(u,v) \in \IL(v)} \pos (u) \right) \\
        &= |\OL(v)| - \zb(v) \cdot (|\IL(v)| + |\OL(v)|). \\
    \end{aligned}
    \end{equation}
    
    Since $\chi(v) = k$, we know that $\IH(v) =\OH(v) = \emptyset$, thus $d(v) = |\IL(v)| + |\OL(v)|$, and $\outdeg(v) = |\OL(v)|$. By this and \cref{eq:yinout-def}, we can see that
    \[
    \begin{aligned}
        \frac{\yout(v) + \zlo(v) - \zli(v)}{\yin(v)+\yout(v)} &= \frac{\alpha \abs*{\OL(v)} + |\OL(v)| - \zb(v) \cdot (|\IL(v)| + |\OL(v)|)}{\alpha \abs*{\OL(v)} + \alpha \abs*{\IL(v)}  } \\
        &= \frac{\alpha+1}{\alpha} \cdot\frac{\outdeg(v)}{d(v)} - \frac{\zb(v)}{\alpha}.
    \end{aligned}
    \]
    Thus, when $\chi(v) = k$, we have $\pos(v) = \tpos(v)$ with probability 1, proving this case.
    
    For the second case where $\chi(v) \neq k$, we will prove this by giving an estimate for the quantities $\abs{\IL(v)}$, $\abs{\OL(v)}$, $\abs{\IH(v)}$, and $\abs{\OH(v)}$ that makes use of the random coloring as mentioned earlier.
    
    \begin{claim}\label{claim:deg-est-bounds}
        For any vertex $v$,
        the following 4 identities hold:

        \[
        \begin{aligned}
            \abs*{\abs{\IL(v)}-\indeg(v) \cdot \frac{\chi(v)-1}{k-1}} &\leq \delta_0 d(v),\\
            \abs*{\abs{\OL(v)}-\outdeg(v) \cdot \frac{\chi(v)-1}{k-1}} &\leq \delta_0 d(v),\\
            \abs*{\abs{\IH(v)}-\indeg(v) \cdot \frac{k-\chi(v)}{k-1}} &\leq \delta_0 d(v),\\
            \abs*{\abs{\OH(v)}-\outdeg(v) \cdot \frac{k-\chi(v)}{k-1}} &\leq \delta_0 d(v).\\
        \end{aligned}
        \]
    \end{claim}
    \begin{proof}
        The proof is a simple triangle inequality from \ref{assumption:balanced-colors}. We see that for a given color $\chi(v)$, there are $\chi(v)-1$ colors less than it, and $k-\chi(v)$ colors greater than it. Since $\abs*{\IL(v)} = \sum_{i=1}^{\chi(v)-1} d_a^-(v)$, we can apply \ref{assumption:balanced-colors}. The others work equivalently.
    \end{proof}
    
    We can now bound some relevant terms. First, we bound $\yin(v)$ and $\yout(v)$. We can use the conditions above and \cref{eq:yinout-def} to see that:
    \begin{equation}\label{yin-highdeg-sub}
        \abs*{\yin(v) - \max \paren*{\indeg(v) \cdot \frac{k-\chi(v)}{k-1}, \alpha \cdot\indeg(v) \cdot \frac{\chi(v)-1}{k-1}}} \leq \delta_0 d(v).
    \end{equation}
    We see that we can factor out a $\indeg(v)$ from the $\max$, and as long as $\chi(v) \neq k$, the first term will always dominate as $\alpha < 1/{k-1}$. This similarly happens for $\yout(v)$ meaning that we have the following 
    \begin{equation}\label{eq:yinout-bound}
    \begin{aligned}
        \abs*{\yin(v)- \indeg(v) \cdot \frac{k-\chi(v)}{k-1}} &\leq \delta_0 d(v), \\
        \abs*{\yout(v)- \outdeg(v) \cdot \frac{k-\chi(v)}{k-1}} &\leq \delta_0 d(v).
    \end{aligned}
    \end{equation}
    And likewise their sum as $\indeg(v) + \outdeg(v) = d(v)$:
    \begin{equation}\label{eq:yinout-sum-bound}
        \abs*{(\yin(v)+\yout(v)) - d(v) \cdot \frac{k-\chi(v)}{k-1}} \leq 2\delta_0 d(v).
    \end{equation}
    
    Now we bound the term $\zlo(v) - \zli(v)$. Again, use the bounds from \cref{claim:deg-est-bounds} that we have conditioned on to substitute each degree estimate into \cref{eq:zlozli-rearranged}. Each of the three substitutions adds an error of $\delta_0 d(v)$. Thus, by triangle inequality and the fact that $\zb(v) \in [0,1]$, we have a total error of less than $3 \delta_0 d(v)$:
    \[
    \abs*{\paren*{\zlo(v) - \zli(v)} - \paren*{\outdeg(v) \cdot \frac{\chi(v)-1}{k-1} - \zb(v) \cdot \paren*{\indeg(v) \cdot \frac{\chi(v)-1}{k-1} + \outdeg(v) \cdot \frac{\chi(v)-1}{k-1}}}} \leq 3\delta_0 d(v),
    \]
    and since $\indeg(v) + \outdeg(v) = d(v)$,
    \begin{equation}\label{eq:zlo-zli-bound}
        \abs*{\paren*{\zlo(v) - \zli(v)} - \paren*{\outdeg(v) \cdot \frac{\chi(v)-1}{k-1} - \zb(v) \cdot d(v) \cdot \frac{\chi(v)-1}{k-1}}} \leq 3\delta_0 d(v).
    \end{equation}
    
    Now, we have bounded additive errors for each term in the formula for $\pos(v)$. For ease of notation, we can express the numerator and denominator for $\pos$ as $X$ and $Y$, and we can do the same for our new substituted estimates as $X'$ and $Y'$, respectively.
    \[
    \begin{aligned}
        X &= \yout(v) + \zlo(v) - \zli(v), \\
        Y &= \yin(v) + \yout(v), \\
        X' &= \outdeg(v) \cdot \frac{k-\chi(v)}{k-1} + \outdeg(v) \cdot \frac{\chi(v)-1}{k-1} - \zb \cdot d(v) \cdot \frac{\chi(v)-1}{k-1}, \\
        Y' &= \indeg(v) \cdot \frac{k-\chi(v)}{k-1} + \outdeg(v) \cdot \frac{k-\chi(v)}{k-1} = d(v) \cdot \frac{k-\chi(v)}{k-1}.
    \end{aligned}
    \]
    As it will be useful later, from this we can also bound $\frac{d(v)}{Y'}$, assuming $\chi(v) \neq k$:
    \begin{equation}\label{eq:dYprime}
        \frac{d(v)}{Y'} = \frac{k-1}{k-\chi(v)} \leq k.
    \end{equation}
    
    We can simplify $\frac{X'}{Y'}$ and see that $\clamp{\frac{X'}{Y'}}$ is equal to $\tpos(v)$.
    \begin{flalign*}
        \clamp{\frac{X'}{Y'}} &= \clamp{\frac{\outdeg(v) \cdot \frac{k-\chi(v)}{k-1} + \outdeg(v) \cdot \frac{\chi(v)-1}{k-1} - \zb \cdot d(v) \cdot \frac{\chi(v)-1}{k-1}}{d(v) \cdot \frac{k-\chi(v)}{k-1}}} \\
        &= \clamp{\frac{k-1}{k-\chi(v)} \cdot \frac{\outdeg(v)}{d} - \zb \cdot \frac{\chi(v)-1}{k-\chi(v)}}\\
        &=  \tpos(v). \tag{By definition \cref{eq:tpos-def} as $\chi(v) \not= k$.)}
    \end{flalign*}
    From \cref{eq:yinout-bound}, \cref{eq:yinout-sum-bound}, and \cref{eq:zlo-zli-bound}, we know that:
    
    \begin{equation}\label{eq:XY-bounds}
    \begin{aligned}
        \abs*{X-X'} &\leq 4\delta_0 d(v), \\
        \abs*{Y-Y'} &\leq 2\delta_0 d(v).
    \end{aligned}
    \end{equation}
    Thus, we have:
    \begin{equation}\label{eq:pos-tpos-start}
    \begin{aligned}
        \abs*{\frac{X}{Y}-\frac{X'}{Y'}} &\leq \abs*{\frac{X}{Y}-\frac{X'}{Y}} + \abs*{\frac{X'}{Y}-\frac{X'}{Y'}} \leq \frac{4\delta_0 d(v)}{Y} + \abs{X'} \cdot \frac{\abs{Y'-Y}}{\abs{Y}\abs{Y'}} \leq \frac{4\delta_0 d(v)}{Y} + \abs*{\frac{X'}{Y'}} \cdot \frac{2\delta_0 d(v)}{Y}.
    \end{aligned}
    \end{equation}
    
    Now, we see that $\abs*{\frac{X'}{Y'}} \leq k$ by examining \cref{eq:tpos-def}: It is a difference of two terms each of which contains a term in the range $[0,1]$ (namely the terms $\frac{\outdeg(v)}{d(v)}$ and $\zb$) and a fractional term which is maximized when $\chi(v) = k-1$ (again, we assume that $\chi(v) \neq k$).
    
    In addition, we must bound $\frac{d(v)}{Y}$ which we do with \cref{eq:XY-bounds} and \cref{eq:dYprime}
    \[
    \frac{d(v)}{Y} \leq \frac{d(v)}{Y'-2\delta_0 d(v)} = \frac{d(v)/Y'}{1-2\delta_0 d(v)/Y'} \leq \frac{k}{1-2\delta_0 k}.
    \]
    We return to \cref{eq:pos-tpos-start}:
    \begin{equation}\label{eq:xy-final-bound}
        \abs*{\frac{X}{Y}-\frac{X'}{Y'}} \leq \frac{4\delta_0 k}{1-2\delta_0 k} + \frac{2\delta_0 k^2}{1- 2\delta_0 k} \leq 3 \delta_0 k^2.
    \end{equation}
    
    Now, we can use \cref{lem:clamp-basic} and substitute $\pos(v)$ and $\tpos(v)$ to conclude the proof of \cref{lem:pos-tpos}
    \[
    \abs*{\pos(v)-\tpos(v)} = \abs*{\clamp{\frac{X}{Y}}-\clamp{\frac{X'}{Y'}}} \leq \abs*{\frac{X}{Y}-\frac{X'}{Y'}} \leq 3 \delta_0 k^2. \qedhere
    \]
\end{proof}

\section{Streaming Adaptation of the Algorithm}

In this section, we describe the streaming algorithm we use to estimate \maxdicut{}. To describe this algorithm, we need a method called Reservoir sampling \cite{vitter1985random}. The goal of reservoir sampling is to uniformly sample a subset of a given size from an initial population in the streaming setting. We write the pseudocode in \Cref{alg:init-reservoir} and \Cref{alg:update-reservoir}.

\begin{algorithm}[H]
    \caption{InitializeReservoir} \label{alg:init-reservoir}
    \KwIn{Integer size $s$}
    \KwOut{Reservoir Data Structure $(S,\mathsf{count})$ containing an array of size $s$ and a counter $\mathsf{count}$}
    Initialize $S \gets $ array of size $s$, $\mathsf{count} \gets 1$. \\
    \KwRet{$(S,\mathsf{count})$}
\end{algorithm}

\begin{algorithm}[H]
    \caption{UpdateReservoir} \label{alg:update-reservoir}
    \KwIn{Reservoir $(S,\mathsf{count})$, and edge $e$}
    \KwOut{Updated Reservoir}
    \If{$\mathsf{count} \leq |S|$}{
        $S[\mathsf{count}] = e$
    }
    \Else{
        With probability $|S|/\mathsf{count}$, select a uniformly random index $i$ in $[1,|S|]$, $S[i] \gets e$.
    }
    $\mathsf{count} \gets \mathsf{count}+1$\\
    \KwRet{$(S,\mathsf{count})$}
\end{algorithm}

For the first $|S|$ updates, the algorithm adds elements to the sample without any condition. After that, the $i^{th}$ element is added to the sample with probability $\frac{|S|}{i}$, replacing one of the previously sampled elements at random. After $i$ steps, the array $S$ contains a uniformly random subset of all elements added. 

We also describe another auxiliary method called the Horvitz-Thompson estimator \cite{horvitz1952generalization} (\Cref{alg:HTAvg}). This is a method of producing an unbiased average over a set of values where each element may fail with some probability $p_i$. The idea is to take a weighted average over the elements that do not fail, normalizing by the probability that we see them. This causes the mean of our $\HTAvg$ to be the same as the true mean. 

\begin{algorithm}
    \caption{Horvitz-Thompson Estimator (HTAvg)}
    \label{alg:HTAvg}
    \KwIn{Takes a list $A$ of values and probabilities $(v_i,p_i)$ or $\bot$ and returns a value.}
    \KwOut{A number which is the Horvitz-Thompson Estimator of the inputs.}
    \If{$A = \emptyset$}{
        \KwRet{$0$}
    }
    \KwRet{\[
    \frac{1}{|A|}\sum_{i=1}^{|A|} v_i/p_i \text{ if $A_i \neq \bot$}
    \]} 
\end{algorithm}

Now we are ready to describe the streaming algorithm (\Cref{alg:streaming-alg}). We start by processing the stream and storing all the information we will need. Then, we have two functions $\VertexEstimator(\cdot)$ and $\EdgeEstimator(\cdot)$ which use this information to estimate the $\pos(\cdot)$ value of a vertex or edge. The streaming part of the algorithm has a few main functions. 

First, we independently sample vertices each with probability $n^{-c}$ from $[n]$\footnote{Our algorithm can be modified to a setting where the vertex set is not $[n]$, but comes from a larger domain $[N]$ where $N = \poly(n)$. To do so, it suffices to use $10\Tmax$-independent hash functions for sampling the vertices.} and add the sampled vertices to a set $W$. For each sampled vertex in $v \in W$, we store a few pieces of information. First, we store degree information for $v$ in the form of 4 counters $\abs*{\IL(v)}, \abs*{\OL(v)}, \abs*{\IH(v)},$ and $\abs*{\OH(v)}$. 

Second, we use reservoir sampling to uniformly sample $d$ incoming edges from $\IL(v)$, and another $d$ outgoing edges from $\OL(v)$. Both of these samplings are done with replacement, and they are stored in $\Rsvsin(v)$ and $\Rsvsout(v)$, respectively (typically reservoir sampling samples without replacement, but by having $d$ reservoirs of size 1, rather than 1 reservoir of size $d$ we can sample with replacement). This is all the information that we store for each sampled vertex in $W$. 

Independently, we also uniformly sample edges with probability $n^{-c}$ in a set $\B$, and we use reservoir sampling to independently sample another set $\C$ of exactly $n^{1-c}$ edges without replacement. The edges in $\B$ will be used to estimate our high-degree vertices, the information relevant to each individual vertex $v \in W$ will be used to estimate our low-degree vertices, and the edges in $\C$ will be used to obtain our final estimate of \maxdicut{}.

\begin{algorithm}
    \caption{Streaming algorithm}
    \label{alg:streaming-alg}
    \KwIn{Graph $G'=(V,E)$ with proper coloring $\chi(v)$ and properties \ref{assumption:linearedges}, \ref{assumption:balanced-in-out}, and \ref{assumption:balanced-colors}.}
    \KwOut{A value $\cutval$ that is an estimate of the maximum directed cut of $G'$} 
    
    \DontPrintSemicolon
    \SetAlgoSkip{bigskip}
    \SetAlgoInsideSkip{}

    \textbf{Pre-processing:}\; 
    
    \Indp
    Initialize Global Sets $W \gets \emptyset$, $\Rsvsin \gets \emptyset$, $\Rsvsout \gets \emptyset$, $\B \gets \emptyset$, $\C \gets \emptyset$ \;

    $\C \gets \textsf{InitializeReservoir}(n^{1-c})$.
    
    \For{every vertex $v \in [n]$}{
        Add $v$ to $W$ independently with probability $n^{-c}$. \\
        \If{$v \in W$}{
            Initialize $\Rsvsin(v), \Rsvsout(v)$ as arrays of size $d$.
            \For{$i$ in $[1,d]$}{
                $\Rsvsin(v)[i] \gets \textsf{InitializeReservoir}(1)$ \\
                $\Rsvsout(v)[i] \gets \textsf{InitializeReservoir}(1)$ \\ 
            }
        }
        If $v$ was added to $W$, then initialize counters for the degree information of $v$, $\abs*{\IL(v)}, \abs*{\OL(v)}, \abs*{\IH(v)}, \abs*{\OH(v)}$. During the stream, we do not store all of these edges, but we can get exact counts for these values with constant space. \label{line:vertex-deg-counters}
        
    }
    \Indm

    \textbf{Stream processing:}
    
    \For{every edge $(u,v)$ in the stream}{
        Add $e = (u, v)$ to $\B$ independently with probability $n^{-c}$\;
        $\textsf{UpdateReservoir}(\C,e)$ \\
        \If{$u \in W$}{
            Update degree counters: $\abs*{\IL(u)}, \abs*{\OL(u)}, \abs*{\IH(u)}, \abs*{\OH(u)}$ as needed. \\
            \If{$\chi(v) < \chi(u)$}{
                \For{$i$ in $[1,d]$}{
                    $\textsf{UpdateReservoir}(\Rsvsout(u)[i],e)$
                }
            }
        }
        \If{$v \in W$}{
            Update degree counters: $\abs*{\IL(v)}, \abs*{\OL(v)}, \abs*{\IH(v)}, \abs*{\OH(v)}$ as needed. \\
            \If{$\chi(u) < \chi(v)$}{
                \For{$i$ in $[1,d]$}{
                    $\textsf{UpdateReservoir}(\Rsvsin(v)[i],e)$
                }
            }
        }
    }
    $\C' \gets \emptyset$ \\
    \For{edge $e \in \C$}{
        \If{$\EdgeEstimator(e) = \bot$}{
            $\C'.\mathsf{append}(\bot)$
        }
        \Else{
            $(\Pest(e),\T(e)) \gets \EdgeEstimator(e)$ \\
            $\C'.\mathsf{append}((\Pest(e),n^{-c|\T(e)|}))$ \label{line:probability-tree-sampling}
        } 
    }
    
    \textbf{Post-processing:}
    
    $\cutval \gets \HTAvg(\C')$\\ \label{line:cutval-def}
    \Return{$\cutval$}
\end{algorithm}

After the stream processing, we call the $\EdgeEstimator(\cdot)$ function on our edges in $\C$ and average them with $\HTAvg(\cdot)$. Now, we describe how our estimators for vertices and edges work. The input to both algorithms is a single edge or vertex with which we plan to estimate its $\pos(\cdot)$ value as $\Pest(\cdot)$. The output, however, is either a pair of our estimate $\Pest(\cdot)$ and a tree $\T(\cdot)$, or simply $\bot$ in the case that we fail. When we get a return value of $(\Pest(\cdot),\T(\cdot))$ for an edge or vertex, this means that our estimate for the true $\pos(\cdot)$ value of that edge or vertex is $\Pest(\cdot)$, and this estimate's success was dependent on $\T(\cdot) \subseteq W$. Thus, since each of these vertices is in $W$ with probability $n^{-c}$, the probability of success is $n^{-c|\T(\cdot)|}$. This is what we pass to our $\HTAvg$.

We describe the edge estimator (\Cref{alg:edge-estimator}). 
For an edge $e = (u,v)$, simply call the vertex estimator on both endpoints. If either one fails, then our edge estimate also fails. If they both succeed, the position of our edge is $\Pest(e) = \Pest(u) \cdot (1-\Pest(v))$, and the result was dependent on the union of $\T(u)$ and $\T(v)$ for success.

\begin{algorithm}[H]
    \caption{Edge Estimator}
    \label{alg:edge-estimator}
    \KwIn{Global: Graph $G'=(V,E)$ with proper coloring $\chi: V \to [k]$ and properties \ref{assumption:linearedges}, \ref{assumption:balanced-in-out}, and \ref{assumption:balanced-colors}, set of vertices $W \subset V$, reservoir data structures $\Rsvsin$ and $\Rsvsout$, set of edges $\B \subset E$}
    \KwOut{$\bot$ or tuple of a position estimate $\Pest(u) \cdot(1-\Pest(v))$, and a set $\T(u,v)$} 
    \textbf{Parameter:} Edge $(u,v)$
    \DontPrintSemicolon
    \SetAlgoSkip{bigskip}
    \SetAlgoInsideSkip{}

    \If{$\VertexEstimator(u) = \bot$ or $\VertexEstimator(v) = \bot$}{
        \KwRet{$\bot$} \\
    }
    $\Pest(u),\T(u) \gets \VertexEstimator(u)$ \\
    $\Pest(v),\T(v) \gets \VertexEstimator(v)$ \\
    \KwRet{$(\Pest(u) \cdot(1-\Pest(v)), \T(u) \cup \T(v))$}

\end{algorithm}

We describe the vertex estimator algorithm (\Cref{alg:vertex-estimator}) which is a bit more complex. The algorithm has two cases for vertices of high and low degree. We determine if a vertex is high-degree using the number of $B$ edges that are adjacent to that vertex (i.e., using $\degB(v)$). If this value is greater than a given threshold based on the color of $v$ (specifically, $n^{q\cdot2^{\chi(v)}}$),  then we consider it high-degree. In this case, we estimate our position value with a formula similar to the formula for $\tpos(\cdot)$ that requires both the average position of lower-colored neighbors and the proportion of outgoing vs incoming adjacent edges to $v$. We can estimate the first by a Horvitz-Thompson average of neighbors along the $B$ edges ($\zbest(v)$), and the second is also easy to estimate using the $\B$ edges. We say that we can give an estimate of $\pos(v)$ for high-degree vertices with probability 1, strictly from the $\B$ edges, thus $\T(v) = \emptyset$ in this case.

In the case of a low-degree vertex $v$, we do not use the edges of $B$. 
Instead, we use the reservoir samples $\Rsvsin(v)$ and $\Rsvsout(v)$. We use a formula similar to $\pos(v)$ that requires degree information, and position estimates from each adjacent lower-colored vertex in $\Rsvsin(v)$ and $\Rsvsout(v)$. Thus, we only succeed if every neighboring vertex in $\Rsvsin(v) \cup \Rsvsout(v)$ succeeds, plus we need $v \in W$. In addition, we also add some dummy vertices to $\T(v)$, each of which fails with the same probability $n^{-c}$. These ensure that the size of $\T(v)$, and thereby the probability of failure, is dependent only on the color of $\chi(v)$. We make sure that $|\T(v)| = \T_{\chi(v)}$. This will help us decouple some of the randomness later on so we can deal with it separately.

Also, we remark that the notation $(u,v) \in \Rsvsin(v)$ is shorthand for enumerating over all edges in the \emph{array} $\Rsvsin(v)$ of size $d$, but we use set notation as it is convenient. In addition, in the low-degree case, our algorithm has access to $\yin(v),\yout(v)$ as we keep counters for $\abs*{\IL(v)}$, $\abs*{\OL(v)}$, $\abs*{\IH(v)}$, and $\abs*{\OH(v)}$ in \Cref{line:vertex-deg-counters} of our streaming algorithm. 

\begin{algorithm}
    \caption{Vertex Estimator}
    \label{alg:vertex-estimator}
    \KwIn{Global: Graph $G=(V,E)$ with coloring $\chi: V \to [k]$ and properties \ref{assumption:linearedges}, \ref{assumption:balanced-in-out}, and \ref{assumption:balanced-colors}, set of vertices $W \subset V$, reservoir data structures $\Rsvsin$ and $\Rsvsout$, set of edges $\B \subset E$}
    \KwOut{$\bot$ or tuple of a position estimate $\Pest(v)$, and a set $\T(v)$} 
    \textbf{Parameter:} Vertex $v$
    \DontPrintSemicolon
    \SetAlgoSkip{bigskip}
    \SetAlgoInsideSkip{}
    
    \If{$\degB(v) > n^{q\cdot2^{\chi(v)}}$}{ 
        \textbf{High Degree Case} \\
        Initialize $\Neighbors(v), \Neighbors'(v)$ as empty lists \\ 
        \For{Neighbor $u$ of $v$ along the edges in $\B$}{
            \If{$\chi(u) < \chi(v)$}{
                $\Neighbors(v)$.append($u$)\\
                \If{$\VertexEstimator(u) = \bot$}{
                    $\Neighbors'(v)$.append($\bot$)\\
                }
                \Else{
                    $(\Pest(u), \T(u)) = \VertexEstimator(u)$\\
                    $\Neighbors'(v)$.append($(\Pest(u), n^{-c|\T(u)|})$)\\
                }
            }
        }
        $\zbest(v) \gets \HTAvg(\Neighbors'(v))$ \\
        $
        \Pest(v) \gets \altclamp \left(
        \begin{cases}
            \frac{\alpha+1}{\alpha} \cdot \frac{\outdegB(v)}{\degB(v)} - \frac{1}{\alpha} \cdot \zbest(v) & \text{if } \chi(v) = k \\
            \frac{k-1}{k-\chi(v)} \cdot \frac{\outdegB(v)}{\degB(v)} - \frac{\chi(v)-1}{k-\chi(v)} \cdot \zbest(v) & \text{if } \chi(v) \neq k  
        \end{cases}
        \right) \label{line:Pest-def-high-deg}
        $\\
        $\T(v) \gets \emptyset$ \\
    }
    \Else {
        \textbf{Low Degree Case} \\
        \If{$v \notin W$}{
            \KwRet{$\bot$} \\
        }
        $\T(v) \gets \set*{v}$ \\
        \For{$(u,v) \in \Rsvsin(v)$ and $(v,u) \in \Rsvsout(v)$}{
            \If{$\VertexEstimator(u) = \bot$}{
                \KwRet{$\bot$} \\
            }
            $\Pest(u),\T(u) = \VertexEstimator(u)$ \\
            $\T(v) \gets \T(v) \cup \T(u)$ \\
        }
        \While{$|\T(v)| < \T_{\chi(v)}$} {
            Create Dummy Vertex $w$, flip coin so $w$ succeeds with probability $n^{-c}$, if so, add it to $W$. \\
            \If{$w \notin W$}{
                \KwRet{$\bot$}
            }
            $\T(v) \gets \T(v) \cup \{w\}$ \tcp*{This ensures $|\T(v)| = \T_{\chi(v)}$}
        }
        $
        \zliest (v) \gets \frac{\abs*{\IL(v)}}{d} \displaystyle \sum_{(u,v) \in \Rsvsin(v)}
        \Pest(u), 
        \hspace{0.5cm}
        \zloest (v) \gets \frac{\abs*{\OL(v)}}{d} \displaystyle \sum_{(v,u) \in \Rsvsout(v)}
        (1-\Pest(u)).
        $\\

        $
        \Pest(v) \gets \clamp{\frac{\yout(v) + \zloest(v) - \zliest(v)}{\yin(v)+\yout(v)}}
        $\\
    }
    \KwRet{$(\Pest(v), \T(v))$}
\end{algorithm}

Our algorithm has several independent sources of randomness that we will need to deal with separately. First, there is the randomness of the edges in $\B$ and in $\C$. Next, there is the randomness in the selections of all reservoirs $\Rsvsin(v)$ and $\Rsvsout(v)$. Lastly, there is the randomness of the sampled vertices in $W$. We deal with each source of randomness separately in \Cref{sec:high-deg-randomness}, \Cref{sec:p1-p2}, and \Cref{sec:p2-p3} respectively.

We now examine the space complexity. We see that by property \Cref{assumption:linearedges}, the total number of edges in $G'$ is at most $10(n/\epsilon^4)$. Each of these edges is sampled with probability $n^{-c}$ in $\B$ giving $\Ex[|\B|] \leq 10n^{1-c}/\epsilon^4 = O(n^c)$. By a Chernoff bound (\Cref{lem:mult-chernoff}), $\Pr \paren*{\abs*{\B} > 20n^{1-c}/\epsilon^4} \leq 2\exp\paren*{-\frac{4}{3}20n^{1-c}\epsilon^4} \leq \epsilon^3$. Similarly, we sample each vertex with probability $n^{-c}$. A similar Chernoff bound gives $\Pr \paren*{\abs*{W} > 2n^{1-c}} \leq \epsilon^3$, not counting the added dummy vertices. Each vertex adds at most constantly many dummy vertices ($\leq \Tmax$) and contributes constant space. Lastly, $|\C| = n^{1-c}$. Thus, with high probability, the algorithm uses $O(n^{1-c})$ space. We can simply fail in the case where either Chernoff bound above goes over the limit, giving us an added failure probability of $2\epsilon^3$.

\section{Revealing Randomness for High-Degree Vertices}\label{sec:high-deg-randomness}
The following claim states that, considering the randomness of $\B$, $d_{\B,a}^+(v) n^c$ and $d_{\B,a}^-(v) n^c$ provide accurate estimations of $d_a^+(v)$ and $d_a^-(v)$, with high probability for high-degree vertices.

\begin{claim}\label{claim:high-deg-bound-prob}
    For any vertex $v$ and color $a$, it holds that
    \[
    \begin{aligned}
        \Pr \paren*{\abs*{d_{\B,a}^+(v) - \outdeg_a(v) \cdot n^{-c}} \geq \delta_0 \cdot \outdeg_a(v) \cdot n^{-c}} &\leq 2 \exp \paren*{-\frac{\delta_0^2}{3} \cdot \outdeg_a(v) \cdot n^{-c}},\\
        \Pr \paren*{\abs*{d_{\B,a}^-(v) - \indeg_a(v) \cdot n^{-c}} \geq \delta_0 \cdot \indeg_a(v) \cdot n^{-c}} &\leq 2 \exp \paren*{-\frac{\delta_0^2}{3} \cdot \indeg_a(v) \cdot n^{-c}}.\\
    \end{aligned}
    \]
\end{claim}
\begin{proof}
    We prove the first as the second is similar. We use the Chernoff bound to estimate the size of $d_{\B,a}^+(v)$ relative to $\outdeg_1(v)$. Define a random variable for each outgoing edge adjacent to $v$ with neighboring color $a$, and let this random variable be equal to 1 if the edge is selected in set $\B$ with probability $n^{-c}$ as in \Cref{alg:streaming-alg}. These random variables are all independent, and their sum is $d_{\B,a}^+(v)$, and the expectation is $\outdeg_a(v) \cdot n^{-c}$. The claim follows from a direct application of the Chernoff bound (\cref{lem:mult-chernoff}).
\end{proof}

Now, for the remainder of the paper, let us condition on these bounds holding for all vertices of high degree and all colors. If a vertex is high-degree, then $d_{a}^+(v) \geq \frac{\epsilon^2d(v)}{2(k-1)} \geq \frac{\epsilon^2n^{2q}}{2(k-1)}$ (by assumption \ref{assumption:balanced-colors}). Therefore, a union bound over all vertices happens with probability at most $4n\exp \paren*{-\frac{\delta_0^2}{3} \cdot \frac{\epsilon^2n^{2q}}{2(k-1)} \cdot n^{-c}} \leq \delta_0$. Thus, with failure probability $\delta_0$, from this point on we fix all randomness of $\B$ and assume that these estimates of $d_{\B,a}^+(v)$ and $d_{\B,a}^-(v)$ are accurate for all high-degree vertices. We note that fixing the randomness of $\B$ also fixes which vertices are high-degree and which are low-degree. This assumption also implies the following by a simple triangle inequality over colors:
\begin{equation} \label{eq:high-deg-estimates}
\begin{aligned}
    \abs*{\outdegB(v) - \outdeg(v) \cdot n^{-c}} &\leq \delta_0 \cdot \outdeg(v) \cdot n^{-c},\\
    \abs*{\degB(v) - d(v) \cdot n^{-c}} &\leq \delta_0 \cdot d(v) \cdot n^{-c}.
\end{aligned}
\end{equation}

With these assumptions, we can now prove two items that will be helpful for the future, the first is a bound on the difference between the fraction of outgoing degree to total degree and the same fraction estimated using the edges in $\B$. The second is a bound on the number of vertices $v$ which can have any given vertex $u$ in its tree $\T(v)$.

\begin{claim} \label{claim:deg-frac-est}
    For any high-degree vertex $v$, it holds that
    \[
    \abs*{\frac{\outdeg(v)}{d(v)} - \frac{\outdegB(v)}{\degB(v)}} \leq 3\delta_0.
    \]
\end{claim}

\begin{proof}
    We begin by manipulating our expression and using \Cref{lem:frac-diff-prelim} and \Cref{eq:high-deg-estimates}.  
    \begin{flalign*}
        \abs*{\frac{\outdeg(v)}{d(v)} - \frac{\outdegB(v)}{\degB(v)}} &= \abs*{\frac{\outdeg(v)\cdot n^{-c}}{d(v)\cdot n^{-c}} - \frac{\outdegB(v)}{\degB(v)}} \\
        &\leq \frac{\abs*{\outdeg(v)\cdot n^{-c} - \outdegB(v)} \cdot d(v) \cdot n^{-c} + \abs*{d(v) \cdot n^{-c} - \degB(v)} \cdot \outdeg(v) \cdot n^{-c}}{d(v) \cdot n^{-c} \cdot \degB(v)} \tag{\cref{lem:frac-diff-prelim}} \\
        &\leq \frac{\delta_0 \cdot \outdeg(v) \cdot n^{-c} \cdot d(v) \cdot n^{-c} + \delta_0 \cdot d(v) \cdot n^{-c} \cdot \outdeg(v) \cdot n^{-c}}{d(v) \cdot n^{-c} \cdot \degB(v)} \\
        &= \frac{2\delta_0 \cdot \outdeg(v) \cdot n^{-c}}{\degB(v)}.
    \end{flalign*}  
    Now, since $\degB(v) \geq (1-\delta_0) \cdot d(v) \cdot n^{-c}$ from \Cref{eq:high-deg-estimates}, and $\frac{\outdeg(v)}{d(v)} \leq 1$, we have
    \[
    \begin{aligned}
        \abs*{\frac{\outdeg(v)}{d(v)} - \frac{\outdegB(v)}{\degB(v)}} &\leq \frac{2\delta_0 \cdot \outdeg(v) \cdot n^{-c}}{\degB(v)} \leq \frac{2\delta_0 \cdot \outdeg(v) \cdot n^{-c}}{(1-\delta_0) \cdot d(v) \cdot n^{-c}} \leq 3\delta_0.
    \end{aligned}
    \]
    This concludes the proof.
\end{proof}

\begin{lemma}\label{lem:count-tree-containment}
    Let $u$ be a vertex of low degree. Then, the total number of vertices $v$ with $\chi(v) \leq a$ and $u \in \T(v)$ is at most $n^{q\cdot 2^{a}}$.
\end{lemma}
\begin{proof}
    We prove the claim by induction. Let $S_b$ be the set of vertices $v$ such that $u \in \T(v)$, and $\chi(v) \leq b$. We note first that $S_{\chi(u)} = 1$ because the only vertex in this set is $u$ itself. This set can only ever contain low-degree vertices as high-degree vertices $v$ have $\T(v) = \emptyset$. Now, we see that a low-degree vertex $v$ can only be in $S_{b+1}$ if it has some child in $S_b$. We can count the maximum number of these by seeing that each element $w$ of $S_b$ can have at most $d(w) \leq  n^{q\cdot 2^b} \cdot n^c/(1-\delta_0)$ neighbors. This is because if $d(w) > n^{q\cdot 2^b} \cdot n^c/(1-\delta_0)$ then by \Cref{eq:high-deg-estimates}, $\degB(w) > n^{q\cdot2^b}$ implying that it is high-degree. However we know that $w$ must be low-degree. Hence $|S_{b+1}| \leq n^{q\cdot 2^b} \cdot n^c \cdot |S_b|/(1-\delta_0)$. This implies the following product.
    \[
    \abs*{S_a} \leq \displaystyle \prod_{b=\chi(u)}^{a-1} n^{q\cdot 2^b} \cdot n^c/(1-\delta_0)\leq n^{q \sum_{b=1}^{a-1} 2^b} \cdot 2^k \cdot n^{ck} \leq  n^{q\cdot(2^{a}-2)} \cdot n^{q} \leq n^{q\cdot 2^{a}}. \qedhere
    \]
\end{proof}


\section{Fixing the Tree Selection} 
\label{sec:fixing-the-trees}
\label{sec:p1-p2}

So far, we have fixed the randomness of $\B$ (which determines the high-degree vertices and the list of their selected neighbors).
In this section, we consider the randomness of the low-degree vertices.
For each \emph{sampled} low-degree vertex $v \in W$, the algorithm selects two sets of $d$ neighbors $\Rsvsin(v)$ and $\Rsvsout(v)$, sampled with replacement from $\IL(v)$ and $\OL(v)$ respectively.
For the sake of the analysis, we imagine the algorithm does the same for \emph{every} low-degree vertex $v$.
Based on the randomness of these selected neighbors, we define a set of intermediate estimations $\Pesth(v)$, for each vertex $v$.
They represent the estimates of our algorithm in a hypothetical situation where every vertex succeeds (i.e., there is no sampling of vertices and everything is included in $W$). We show that for each vertex $v$, $\Ex[\val{G'}{\Pesth}]$ is close to $\val{G'}{\pos}$. Later, we show that after fixing a tree selection, and hence fixing $\val{G'}{\Pesth}$, the cut value corresponding to the estimated positions $\Pest$ (i.e., $\Ex[\val{G'}{\Pest}]$) is close to $\val{G'}{\Pesth}$.

Formally, for each low-degree vertex $v$, we let
$$
    \zliesth (v) = \frac{\abs*{\IL(v)}}{d} \displaystyle \sum_{(u,v) \in \Rsvsin(v)}
    \Pesth(u), 
    \hspace{0.5cm}
    \zloesth (v) = \frac{\abs*{\OL(v)}}{d} \displaystyle \sum_{(v,u) \in \Rsvsout(v)}
    (1-\Pesth(u)),
$$
and
\begin{equation}
    \Pesth(v) = \clamp{\frac{\yout(v) + \zloesth(v) - \zliesth(v)}{\yin(v)+\yout(v)}}. \label{eq:lo-pesth}
\end{equation}
For each high-degree vertex $v$, we let
$$
\zbesth(v) = \frac{1}{\card{\Neighbors(v)}} \sum_{u \in \Neighbors(v)} \Pesth(u),
$$
and
\begin{equation}
        \Pesth(v) = \altclamp \left(
        \begin{cases}
            \frac{\alpha+1}{\alpha} \cdot \frac{\outdegB(v)}{\degB(v)} - \frac{1}{\alpha} \cdot \zbesth(v) & \text{if } \chi(v) = k \\
            \frac{k-1}{k-\chi(v)} \cdot \frac{\outdegB(v)}{\degB(v)} - \frac{\chi(v)-1}{k-\chi(v)} \cdot \zbesth(v) & \text{if } \chi(v) \neq k  
        \end{cases}.
        \right) \label{eq:hi-pesth}
\end{equation}
Notably, the estimates above differ from $\Pest(v)$ in that they are defined for all vertices $v$, and there is no rescaling (as in the Horvitz-Thompson average).

\begin{lemma}\label{lem:pesth-variance}
    For any vertex $v$ with $\chi(v) = a$, it holds that
        $$ \Var(\Pesth(v)) \leq \sigma_a^2, $$
    where the randomness is over the selection of the low-degree neighbors.
\end{lemma}
\begin{proof}
    We prove the claim by induction on $a$.
    For the base case, $a = 1$, the value $\Pesth(v)$ is a constant as it depends only on the degree information of $v$.
    That is, $\Pesth(v)$ is determined by $\yin(v)$ and $\yout(v)$ for low-degree vertices, which are absolute constants, and by $\outdegB(v)$ and $\degB(v)$  for high-degree vertices, which are constants given that we have fixed the randomness of $\B$.
    Therefore, $\Pesth(v)$ is a constant, and $\Var(\Pesth(v))$ is equal to $0$.
    
    For the induction step, first, we examine a high-degree vertex $v$.
    Recall the definition of $\Pesth(v)$ for high-degree vertices \eqref{eq:hi-pesth}.
    We can remove the clamp since doing so can only increase the variance (\cref{lem:var-of-clamp}).
    Observe that the terms involving $\outdegB$ and $\degB$ are constants since we have fixed the randomness of $\B$.
    Also for any color $\chi(v)$ it holds $\frac{\chi(v) - 1}{k - \chi(v)} \leq k - 2 \leq \frac{1}{\alpha}$.
    Therefore, we have
    \begin{equation}
    \Var(\Pesth(v)) \leq \Var\left(\frac{1}{\alpha} \zbesth(v) \right) = \frac{1}{\alpha^2} \Var(\zbesth(v)). \label{eq:hi-pesth-ub}
    \end{equation}
    To bound $\Var(\zbesth(v))$, we expand it as follows:
    \begin{align*}
        \Var(\zbesth) 
        &= \frac{1}{\card{\Neighbors(v)}^2}\Var\left(\sum_{u \in \Neighbors(v)} \Pesth(u) \right) \\
        &= \frac{1}{\card{\Neighbors(v)}^2}\sum_{u, w \in \Neighbors(v)} \Cov(\Pesth(u), \Pesth(w)) \\
        &\leq \frac{1}{\card{\Neighbors(v)}^2}\sum_{u, w \in \Neighbors(v)} \sqrt{\Var(\Pesth(u)) \Var(\Pesth(w))}  \tag{by \cref{lem:cauchy-schwarz}}\\
        &\leq \frac{1}{\card{\Neighbors(v)}^2} \cdot \card{\Neighbors(v)}^2\sigma_{a-1}^2 \tag{by the induction hypothesis}\\
        &= \sigma_{a-1}^2.
    \end{align*}
    Plugging this back into \eqref{eq:hi-pesth-ub}, we obtain:
    $$
    \Var(\Pesth(v)) \leq \frac{\sigma_{a-1}^2}{\alpha^2} \leq \sigma_a^2.
    $$

    The analysis is slightly more complicated for low-degree vertices.
    We start by removing the clamp (\cref{lem:var-of-clamp}) and the constants similar to the previous case:
    \begin{equation}
    \Var(\Pesth(v)) \leq \frac{1}{(\yin(v) + \yout(v))^2}\Var(\zloesth(v) - \zliesth(v))
    \leq \frac{1}{\alpha^2(\card{\IL} + \card{\OL})^2}\Var(\zloesth(v) - \zliesth(v)).
    \label{eq:lo-pesth-ub}
    \end{equation}
    Let $X_1, \ldots, X_d$ denote the estimate $\Pesth$ for the selected neighbors of $v$ in $\Rsvsin(v)$, and let $X'_1, \ldots, X'_d$ denote the estimate for the vertices in $\Rsvsout(v)$.
    Then, we can expand $\Var(\zloesth(v) - \zliesth(v))$ as follows:
    \begin{align}
        \Var(\zloesth(v) &- \zliesth(v))
        =  \Var\left( \frac{\abs*{\OL(v)}}{d} \sum_{i} 1 - X'_i - \frac{\abs*{\IL(v)}}{d} \sum_i X_i \right) \notag \\
        &= \frac{\card{\OL(v)}^2}{d^2}\sum_{i, j} \Cov(X'_i, X'_j) 
        + 2\frac{\card{\OL(v)}\card{\IL(v)}}{d^2}\sum_{i, j} \Cov(X'_i, X_j) \notag \\
        &\quad + \frac{\card{\IL(v)}^2}{d^2}\sum_{i, j} \Cov(X_i, X_j). \label{eq:pesth-covars}
    \end{align}

    We bound the first in \eqref{eq:pesth-covars} sum as follows:
    $$
    \sum_{i, j} \Cov(X'_i, X'_j)
    = \sum_i \Var(X'_i) + \sum_{i \neq j} \Cov(X'_i, X'_j)
    \leq \frac{d}{4} + d^2 \sigma_{a-1}^2,
    $$
    where for the last inequality, we trivially bound the variance of $X_i'$ by $\frac{1}{4}$ (as it holds $X_i' \in [0, 1]$), and bound $\Cov(X'_i, X'_j)$ for $i \neq j$ using the induction hypothesis and \cref{lem:sampling-covariance}.
    Similarly, it holds:
    $$
    \sum_{i, j} \Cov(X_i, X_j) \leq \frac{d}{4} + d^2\sigma_{a-1}^2.
    $$
    Lastly, for the covariances between $X_i'$ and $X_j$ we can use \cref{lem:sampling-covariance} alone, since there are no variances involved:
    $$
    \sum_{i, j} \Cov(X'_i, X_j) \leq d^2 \sigma_{a-1}^2.
    $$
    Plugging these bounds back into \eqref{eq:pesth-covars}, we obtain:
    \begin{align*}
        \Var(\zloesth(v) - \zliesth(v))
        &\leq \frac{\card{\OL(v)}^2}{d^2}\left(\frac{d}{4} + d^2\sigma_{a-1}^2\right)
        + \frac{2\card{\OL(v)}\card{\IL(v)}}{d^2} d^2 \sigma_{a-1}^2 
        + \frac{\card{\IL(v)}^2}{d^2}\left(\frac{d}{4} + d^2\sigma_{a-1}^2\right) \\
        &\leq (\card{\IL(v)} + \card{\OL(v)})^2 \left(\sigma_{a-1}^2 + \frac{1}{4d}\right).
    \end{align*}
    Plugging this back into \eqref{eq:lo-pesth-ub}, we get
    $$
    \Var(\Pesth(v)) \leq \frac{\sigma_{a-1}^2 + 1/4d}{\alpha^2} \leq \sigma_a^2,
    $$
    which concludes the proof of the induction.
\end{proof}

\begin{lemma}\label{lem:pesth-mean}
    For any vertex $v$ with $\chi(v) = a$, it holds that
        $$ \card{\Ex[\Pesth(v)] - \pos(v)} \leq \delta_a, $$
    where the randomness is over the selection of the low-degree neighbors.
\end{lemma}
\begin{proof}
    We prove this inductively. First, we see that if $a = 1$, then $\Pesth(v) = \pos(v)$ if $v$ is low-degree and $\Pesth(v) = \tpos(v)$ if $v$ is high-degree with probability 1. Thus the lemma follows from \Cref{lem:pos-tpos}. Now, we assume inductively that for any vertex $u$ with $\chi(u) < a$, we have
    \[
    \card{\Ex[\Pesth(u)] - \pos(u)} \leq \delta_{a-1}.
    \]
    Begin with the case where $v$ is low-degree. We expand the formulas for $\Pesth(v)$ and $\pos(v)$.
    \begin{align}
        \card{\Ex[\Pesth(v)] - \pos(v)} &= \abs*{\Ex \bracket*{\clamp{\frac{\yout(v) + \zloesth(v) - \zliesth(v)}{\yin(v)+\yout(v)}}} - \clamp{\frac{\yout(v) + \zlo(v) - \zli(v)}{\yin(v)+\yout(v)}}} \notag \\
        &\leq \abs*{\Ex \bracket*{\frac{\yout(v) + \zloesth(v) - \zliesth(v)}{\yin(v)+\yout(v)}} - \frac{\yout(v) + \zlo(v) - \zli(v)}{\yin(v)+\yout(v)}} + \sqrt{\Var(\Pesth(v))} \tag{by \Cref{lem:conditional-clamping}} \\
        &\leq \frac{1}{\alpha \cdot (\abs*{\IL(v)} + \abs*{\OL(v)})} \paren*{\abs*{\Ex [\zloesth(v)] - \zlo(v)} + \abs*{\Ex [\zliesth(v)] - \zli(v)}} + \sqrt{\Var(\Pesth(v))}. \label{eq:pesth-pos-mean-1}
    \end{align}
    From \Cref{lem:pesth-variance}, we know that $\sqrt{\Var(\Pesth(v))} \leq \sigma_a$, and so we can focus on $\abs*{\Ex [\zliesth(v)] - \zli(v)}$ and $\abs*{\Ex [\zloesth(v)] - \zlo(v)}$. We only prove a bound for the first as the second is analogous. To do this, first note that if $\IL(v)$ is empty then the bound is trivial. Otherwise, we must figure out what $\Ex[\zliesth(v)]$ is. $\zliesth(v)$ is a sum of $d$ neighboring $\Pesth(\cdot)$ values selected uniformly at random with replacement. Let these values be the random variables $X_1, \dots X_d$. Then $\Ex[\zliesth(v)] = \Ex \bracket*{\frac{\abs*{\IL(v)}}{d}\sum_{i=1}^d X_i}$. Each $X_i$ has identical expectation and it is equal to the average $\Ex[\Pesth(\cdot)]$ value of incoming lower-colored vertices as it takes each of these values with equal probability. Thus we have 
    \[
    \Ex[\zliesth(v)] = \frac{\abs*{\IL(v)}}{d} \cdot \sum_{i=1}^d \paren*{\frac{1}{\abs*{\IL(v)}} \sum_{(u,v) \in \IL(v)} \Ex[\Pesth(u)]} = \sum_{(u,v) \in \IL(v)} \Ex[\Pesth(u)].
    \]
    We can now we can use this to bound the difference between $\Ex[\zliesth(v)]$ and $\zli(v)$.
    \[
    \begin{aligned}
        \abs*{\Ex [\zliesth(v)] - \zli(v)} &= \abs*{\sum_{(u,v) \in \IL(v)} \Ex[\Pesth(u)] - \sum_{(u,v) \in \IL(v)} \pos(v)} \\
        &\leq \sum_{(u,v) \in \IL(v)} \abs*{\Ex[\Pesth(u)]-\pos(v)} \leq \abs*{\IL(v)} \cdot \delta_{a-1}.
    \end{aligned}
    \]
    With the last step coming from applying the inductive hypothesis. We similarly find that $ \abs*{\Ex [\zloesth(v)] - \zlo(v)} \leq \abs*{\OL(v)} \cdot \delta_{a-1}$ and we can plug both bounds as well as \Cref{lem:pesth-variance} back into \Cref{eq:pesth-pos-mean-1} to get the following.
    \[
    \begin{aligned}
        \card{\Ex[\Pesth(v)] - \pos(v)} &\leq \frac{1}{\alpha \cdot (\abs*{\IL(v)} + \abs*{\OL(v)})} \paren*{\abs*{\OL(v)} \cdot \delta_{a-1} + \abs*{\IL(v)} \cdot \delta_{a-1}} + \sigma_a \\
        &\leq \frac{\delta_{a-1}}{\alpha} + \sigma_a \leq \delta_a.
    \end{aligned}
    \]
    This concludes the proof for low-degree vertices. 
    
    We now consider the case where $v$ is high-degree. In this case, we will prove a bound for $\card{\Ex[\Pesth(v)]-\tpos(v)}$ as we know that $\tpos(v)$ and $\pos(v)$ are close. Similarly to before, we can expand the formulas for $\Pesth(v)$ and $\tpos(v)$. At the same time, we use \Cref{lem:conditional-clamping} to remove the clamps. Also, we will use the case where $\chi(v) = k$ for the formula as it is an upper-bound on both cases because $\frac{\chi(v) - 1}{k - \chi(v)} \leq k - 2 \leq \frac{1}{\alpha}$.
    \[
    \begin{aligned}
        \abs*{\Ex [\Pesth(v)] - \tpos(v)} &\leq \abs*{\Ex \bracket*{\frac{\alpha+1}{\alpha} \cdot \frac{\outdegB(v)}{\degB(v)} - \frac{1}{\alpha} \cdot \zbesth(v)} - \paren*{\frac{\alpha+1}{\alpha} \cdot \frac{\outdeg(v)}{d(v)} - \frac{1}{\alpha} \cdot \zb(v)}} + \sqrt{\Var(\Pesth(v))} \\
        &\leq \frac{\alpha+1}{\alpha} \cdot \abs*{\frac{\outdegB(v)}{\degB(v)} - \frac{\outdeg(v)}{d(v)}} + \frac{1}{\alpha}\cdot \abs*{\Ex [\zbesth(v)] - \zb(v)} + \sqrt{\Var(\Pesth(v))}\\
        &\leq \frac{\alpha+1}{\alpha} \cdot 3\delta_0 + \frac{1}{\alpha}\cdot \abs*{\Ex \bracket*{\frac{1}{\abs*{\Neighbors(v)}}\displaystyle \sum_{u \in \Neighbors(v)} \Pesth(u)} - \frac{1}{\abs*{\Neighbors(v)}}\displaystyle \sum_{u \in \Neighbors(v)} \pos(u)} + \sqrt{\Var(\Pesth(v))}\\
        &\leq \frac{\alpha+1}{\alpha} \cdot 3\delta_0 + \frac{1}{\alpha\abs*{\Neighbors(v)}}\cdot \displaystyle \sum_{u \in \Neighbors(v)} \abs*{\Ex \bracket*{\Pesth(u)} - \pos(u)} + \sqrt{\Var(\Pesth(v))} \\
        &\leq \frac{\alpha+1}{\alpha} \cdot 3\delta_0 + \frac{1}{\alpha} \cdot \delta_{a-1} + \sigma_a \leq \delta_a/2.
    \end{aligned}
    \]
    Where the final step comes from \Cref{lem:pesth-variance} and the inductive hypothesis. Combining this with \Cref{lem:pos-tpos} proves the claim.
\end{proof}

Now that we have bounded the variance of $\Pesth(\cdot)$ of a vertex and the distance from its mean to its $\pos(\cdot)$, we can now say that the $\Pesth(\cdot)$ of a vertex is close to $\pos(\cdot)$ with high probability. Then, we extend the result to edges.

\begin{lemma}\label{lem:pesth-pos-prob}
    For any vertex $v$, it holds
    \[
    \Pr \paren*{\abs*{\Pesth(v)-\pos(v)} > 2\delta_k} \leq \delta_k.
    \]
\end{lemma}
\begin{proof}
    By Chebyshev (\Cref{lem:chebyshev}) and \Cref{lem:pesth-variance}, it holds that 
    \[
    \Pr \paren*{\abs*{\Pesth(v)-\Ex[\Pesth(v)]} > \delta_k} \leq \frac{\sigma_k^2}{\delta_k^2} \leq \delta_k.
    \]
    Thus, combining this with \Cref{lem:pesth-mean} proves the claim with a simple triangle inequality.
\end{proof}

\begin{lemma}\label{lem:pesth-pos-prob-edge}
    For any edge $(u,v)$, it holds
    \[
    \Pr \paren*{\abs*{\Pesth(u)\cdot(1-\Pesth(v))-\pos(u)\cdot(1-\pos(v))} > 6\delta_k} \leq 2\delta_k.
    \]
\end{lemma}
\begin{proof}
    By a union bound on \Cref{lem:pesth-pos-prob} for vertices $u$ and $v$, it suffices to show that if we condition on the events where $\abs*{\Pesth(v)-\pos(v)} \leq 2\delta_k$ and $\abs*{\Pesth(u)-\pos(u)} \leq 2\delta_k$, then with probability 1 we have
    \[
    \abs*{\Pesth(u)\cdot(1-\Pesth(v))-\pos(u)\cdot(1-\pos(v))} \leq 6\delta_k.
    \]
    We show this as follows.
    \begin{align*}
        \abs*{\Pesth(u)\cdot(1-\Pesth(v))-\pos(u)\cdot(1-\pos(v))} &\leq \abs*{\Pesth(u)-\pos(u)} + \abs*{\Pesth(u)\Pesth(v)-\pos(u)\pos(v)} \\
        &\leq 2\delta_k + \abs*{\Pesth(u)\Pesth(v)-\Pesth(u)\pos(v)} + \abs*{\Pesth(u)\pos(v)-\pos(u)\pos(v)} \\
        &= 2\delta_k + \Pesth(u)\cdot\abs*{\Pesth(v)-\pos(v)} + \pos(v)\cdot\abs*{\Pesth(u)-\pos(u)}\\
        &\leq 2\delta_k + 2\delta_k + 2\delta_k = 6\delta_k. \qedhere
    \end{align*}
\end{proof}

With this, we can now show that $\val{G'}{\Pesth}$ is close to $\val{G'}{\pos}$ with high probability. This allows us to fix the randomness of $\Rsvsin(\cdot)$ and $\Rsvsout(\cdot)$ with small failure probability.

\begin{lemma}\label{lem:valG-Pesth-valG-pos}
    \[
    \Pr(\abs*{\val{G'}{\Pesth} - \val{G'}{\pos}} > 13\delta_k) \leq \delta_k.
    \]
\end{lemma}
\begin{proof}
    The strategy is to show that the mean of $\val{G'}{\Pesth}$ is close to $\val{G'}{\pos}$, then we bound the variance and use \Cref{lem:chebyshev}.
    \[
    \begin{aligned}
        \abs*{\Ex\bracket*{\val{G'}{\Pesth}} - \val{G'}{\pos}} &\leq \abs*{\Ex \bracket*{\frac{1}{|E|} \sum_{(u,v) \in E} \Pesth(u) \cdot (1-\Pesth(v))} - \frac{1}{|E|} \sum_{(u,v)\in E} \pos(u) \cdot (1-\pos(v))} \\
        &\leq \frac{1}{|E|} \sum_{(u,v)\in E} \abs*{\Ex \bracket*{\Pesth(u) \cdot (1-\Pesth(v))} - \pos(u) \cdot (1-\pos(v))}.
    \end{aligned}
    \]

    From \Cref{lem:pesth-pos-prob} and \Cref{lem:custom-prob-ex}, we know that $\abs*{\Ex \bracket*{\Pesth(u) \cdot (1-\Pesth(v))} - \pos(u) \cdot (1-\pos(v))} \leq 12\delta_k$. This gives us
    \begin{equation}\label{eq:mean-valG-Pesth}
        \abs*{\Ex\bracket*{\val{G'}{\Pesth}} - \val{G'}{\pos}} \leq 12\delta_k.
    \end{equation}

    Now, we bound the variance of $\val{G'}{\Pesth}$ as follows, notating that for an edge $e = (u,v)$ $\Pesth(e) = \Pesth(u) \cdot (1-\Pesth(v))$.
    \[
    \begin{aligned}
        \Var \paren*{\val{G'}{\Pesth}} &= \Var \paren*{\frac{1}{|E|} \sum_{e \in E} \Pesth(e)} \\
        &= \frac{1}{|E|^2} \sum_{e_1 \in E} \sum_{e_2 \in E} \Cov \paren*{\Pesth(e_1), \Pesth(e_2)}.
    \end{aligned}
    \]

    By \Cref{lem:cauchy-schwarz}, $\Cov \paren*{\Pesth(e_1), \Pesth(e_2)} \leq \sqrt{\Var(\Pesth(e_1))\Var(\Pesth(e_2))}$. We can bound the variance of $\Pesth(e_1) = \Pesth(u) \cdot (1-\Pesth(v))$ using \Cref{lem:pesth-variance} and \Cref{lem:var-product} as at most $9\sigma_k^2$. Thus, this bound carries to our overall variance.
    \[
    \Var \paren*{\val{G'}{\Pesth}} \leq 9\sigma_k^2.
    \]
    Now, we can use \Cref{lem:chebyshev} to see that
    \[
    \Pr \paren*{\abs*{\val{G'}{\Pesth} - \Ex \bracket*{\val{G'}{\Pesth}}} > \delta_k} \leq \frac{9\sigma_k^2}{\delta_k^2} \leq \delta_k.
    \]
    Thus, the lemma follows from a triangle inequality with this and \Cref{eq:mean-valG-Pesth}.
\end{proof}

\section{Dealing with Vertex Sampling}
\label{sec:p2-p3}

Up until this point, we have not dealt with the randomness of the vertex sampling in $W$, nor have we dealt with the scaling in our Horvitz-Thompson Estimator (\Cref{alg:HTAvg}). This section addresses this, as the only randomness remaining in the algorithm is this vertex sampling in $W$, and the random sampling of $\C$, which we will deal with later as it is independent from the rest of the algorithm. In our algorithm, for each vertex $v$ we keep track of two items, $\Pest(v)$ and $\T(v)$. $\Pest(v)$ is our position estimate, but our algorithm is only able to compute this if we sample each vertex in $\T(v)$. Otherwise, the vertex estimator fails. Our edge-estimator works similarly, with both endpoints needing their respective trees to be sampled. We define the indicator variable for a vertex $v$ to be $X_v = \Pest(v) \cdot n^{c\abs*{\T(v)}}$ when $\T(v) \subseteq W$, and 0 otherwise. We note that this is the value used in our Horvitz-Thompson Estimator (\Cref{alg:HTAvg}). Similarly for an edge $e=(u,v)$, the indicator variable $X_e = \Pest(e) \cdot n^{c\abs*{\T(u) \cup \T(v)}}$. The key property of these indicator variables which is necessary to note is that $\Ex[X_v] = \Ex[\Pest(v)\mid\mc{E}_v]$ where $\mc{E}_v$ is the event where $\T(v) \subseteq W$. This is because the scaling up from the formula for $X_v$ by $n^{c\abs*{\T(v)}}$ cancels out with $\Pr(\mc{E}_v) = n^{-c\abs*{\T(v)}}$ as we sample each vertex independently with probability $n^{-c}$.

The goal of this section is similar to the previous one, we want to show that $\val{G'}{\Pest}$ is close to $\val{G'}{\Pesth}$ with good probability. To do this, as before, we must bound the variance and show that our mean is close, then we can use Chebyshev (\Cref{lem:chebyshev}). As before, we do this for individual vertices first. The main difference, however, is now it is possible to fail an estimate for a vertex, so we can only bound the expectation conditioned on that vertex's success. 

In order to prove a variance bound for $\Pest(\cdot)$, we will prove a stronger lemma (\Cref{lem:variance+mean-shift}) which bounds both a more general conditioning on variance and a difference in conditional expectation. We do this to strengthen our inductive hypothesis. The crux of the proof is that vertices with disjoint trees from one another have little affect on each other's position estimate, though the existence of high-degree vertices implies that there may be a large number of such vertices which have small non-zero covariance with any given vertex. Because of this, we define a slowly-increasing function $f(x)$ to represent how $x$ disjoint vertices affect the mean and variance of $\Pest(\cdot)$. 

Before stating our lemma, we define some notation. Define the functions $f(x)$ and $g(x)$ as follows.

\begin{align}
    f(x) &= 1+\frac{x}{n^q} \label{eq:f-def}, \\
    g(x) &= f(x)^2 \label{eq:g-def}.
\end{align}

We note some properties of $f(x)$ and $g(x)$ on $x \geq 0$. Namely that they are both increasing, at least 1, and $g(x) \geq f(x)$.

\begin{lemma}\label{lem:variance+mean-shift}
    Take a vertex $v$ with $\chi(v) = a$, and vertex sets $R$ and $S$. Let $\mc{E}_v$, $\mc{E}_R$, and $\mc{E}_S$ be the events that $T(v) \subseteq W$, $R \subseteq W$, and $S \subseteq W$ respectively. The following statements hold:
    \begin{enumerate}
        \item $\Var \paren*{\Pest(v) \mid \mc{E}_v \cap \mc{E}_S} \leq \sigma_a^2 \cdot g(|S|)$, and
        \item $\card{\Ex[\Pest(v) \mid \mc{E}_v \cap \mc{E}_R]- \Ex[\Pest(v) \mid \mc{E}_v \cap \mc{E}_R \cap \mc{E}_S]} \leq 3\sigma_{a} \cdot f(\card{R}+ \card{S}).$
    \end{enumerate}
\end{lemma}

\begin{proof}
    We first note that without loss of generality we can assume that $\T(v), R$, and $S$ are pairwise disjoint sets. If they aren't, we simply construct $R' = R-\T(v)$ and $S' = S-(\T(v) \cup R)$, along with their analogous events $\mc{E}_{R'}, \mc{E}_{S'}$. It is clear that the sets $\T(v), R', S'$ are pairwise disjoint. and if we assume the lemma holds in this case, we can see the following:
    \[
    \begin{aligned}
        \Var \paren*{\Pest(v) \mid \mc{E}_v \cap \mc{E}_S} &= \Var \paren*{\Pest(v) \mid \mc{E}_v \cap \mc{E}_{S'}} \leq \sigma_a^2 \cdot g(|S'|) \leq \sigma_a^2 \cdot g(|S|),
    \end{aligned}
    \]
    \[
    \begin{aligned}
        \card{\Ex[\Pest(v) \mid \mc{E}_v \cap \mc{E}_R]- \Ex[\Pest(v) \mid \mc{E}_v \cap \mc{E}_R \cap \mc{E}_S]} &= \card{\Ex[\Pest(v) \mid \mc{E}_v \cap \mc{E}_{R'}]- \Ex[\Pest(v) \mid \mc{E}_v \cap \mc{E}_{R'} \cap \mc{E}_{S'}]} \\
        &\leq 3\sigma_{a} \cdot f(\card{R'}+ \card{S'}) \leq 3\sigma_{a} \cdot f(\card{R}+ \card{S}).
    \end{aligned}
    \]

    Thus, from this point, we assume that $\T(v), R, S$ are pairwise disjoint. We prove this claim by induction. In the base case where $a=1$, once we condition on $\mc{E}_v$, our value of $\Pest(v)$ is determined. This is because all randomness in high-degree vertices can come only from the $\zbest(v)$ term, which is zero. For low-degree vertices, the only randomness comes from the $\zloest(v)$ and $\zliest(v)$ terms which are similarly zero, and since we condition on $\mc{E}_v$ $(\T(v) = \{v\})$ in both equations, there is no variance and $\Pest(v)$ is constant making the left hand sides of both equations zero. Thus, the lemma follows in the base case.

    The inductive case, we prove the statements in order, assuming that both claims hold for all vertices with color less than $a$. This gives us the following two inductive hypotheses for any vertex $u$ with $\chi(u) < a$ and any vertex sets $R$ and $S$.

    \begin{equation}\label{eq:variance-ih}
        \Var \paren*{\Pest(u) \mid \mc{E}_u \cap \mc{E}_S} \leq \sigma_{a-1}^2 \cdot g(|S|).
    \end{equation}
    \begin{equation}\label{eq:mean-shift-ih}
        \card{\Ex[\Pest(u) \mid \mc{E}_u \cap \mc{E}_R]- \Ex[\Pest(v) \mid \mc{E}_u \cap \mc{E}_R \cap \mc{E}_S]} \leq 3\sigma_{a-1} \cdot f(\card{R}+ \card{S}).
    \end{equation}

    Assuming these hypotheses, we prove the two claims in order as we must assume the first to prove the second.
    
    \begin{claim}\label{claim:conditional-variance}
        Let $v$ be a vertex with $\chi(v) = a$. Let $S$ be any set of low degree vertices disjoint from $\T(v)$. Let $\mc{E}_v$ be the event where $\T(v) \subseteq W$, and $\mc{E}_S$ be the event where $S \subseteq W$. Then we have
        \[
            \Var \paren*{\Pest(v) \mid \mc{E}_v \cap \mc{E}_S} \leq \sigma_a^2 \cdot g(|S|).
        \]
    \end{claim}

    \begin{proof}
        To prove this, we split into cases for high- and low-degree. Beginning with high-degree, $\T(v) = \emptyset$, so $\mc{E}_v$ happens with probability 1, and we can ignore it.
        \[
        \Var(\Pest(v)\mid\mc{E}_S) = \Var \left( \altclamp \left(
            \begin{cases}
                \frac{\alpha+1}{\alpha} \cdot \frac{\outdegB(v)}{\degB(v)} - \frac{1}{\alpha} \cdot \zbest(v) & \text{if } \chi(v) = k \\
                \frac{k-1}{k-\chi(v)} \cdot \frac{\outdegB(v)}{\degB(v)} - \frac{\chi(v)-1}{k-\chi(v)} \cdot \zbest(v) & \text{if } \chi(v) \neq k  
            \end{cases}
            \right)\ \middle|\ \mc{E}_S\right).
        \]
        
        Now, we perform several steps at once. First, use \Cref{lem:var-of-clamp} to get rid of the clamp. Then, in both cases the first term is constant and can be removed. We can also negate the second term, and in both cases bring out the constant in front of $\zbest$, but since $1/\alpha$ is always larger, we use that as the upper bound. This gives us the following.
        \[
        \Var(\Pest(v)\mid\mc{E}_S) \leq \frac{1}{\alpha^2} \Var \paren*{\zbest(v)\mid\mc{E}_S}.
        \]
        
        We now plug in the formula for $\zbest(v)$, and let $u_1, \dots, u_{\abs*{\Neighbors(v)}}$ enumerate the elements of $\Neighbors(v)$, and let their corresponding indicator variables be $X_1, \dots, X_{\abs*{\Neighbors(v)}}$ respectively. 
        \[
        \Var(\Pest(v)\mid\mc{E}_S) \leq \frac{1}{\alpha^2} \Var \paren*{\frac{1}{\abs*{\Neighbors(v)}} \displaystyle \sum_{i=1}^{\abs*{\Neighbors(v)}} X_i\ \middle|\ \mc{E}_S}.
        \]
        
        Next, we pull out the $\frac{1}{\abs*{\Neighbors(v)}}$ and use the fact that the variance of a sum is the sum of pairwise covariances.
        \[
        \Var(\Pest(v)\mid\mc{E}_S) \leq \frac{1}{\alpha^2\abs*{\Neighbors(v)}^2} \displaystyle \sum_{i=1}^{\abs*{\Neighbors(v)}} \displaystyle \sum_{j=1}^{\abs*{\Neighbors(v)}} \Cov \paren*{X_i, X_j \mid \mc{E}_S}.
        \]
        
        We now split this sum into good and bad pairs of $(i,j)$. Let $(i,j)$ be good if $\T(u_i), \T(u_j), S$ are pairwise disjoint sets. Otherwise, call the pair $(i,j)$ bad. We will prove that if $(i,j)$ is a good pair, then their covariance is very small , and if it is a bad pair, the covariance may be very large, though they only make up a small fraction of the total number of pairs. Splitting up the sum in this way gives us the following.
        \begin{equation} \label{eq:cov-split-good-bad}
            \Var(\Pest(v)\mid\mc{E}_S) \leq \frac{1}{\alpha^2\abs*{\Neighbors(v)}^2} \paren*{\mathop{\sum\sum}_{(i,j) \text{ good}} \Cov \paren*{X_i, X_j \mid \mc{E}_S} + \mathop{\sum\sum}_{(i,j) \text{ bad}} \Cov \paren*{X_i, X_j \mid \mc{E}_S}}.
        \end{equation}

         Now, assume $(i,j)$ is a good pair. We split up the covariance and use chain-rule on the events $\mc{E}_i$ and $\mc{E}_j$ that correspond to $\T(u_i) \subseteq W$ and $\T(u_j) \subseteq W$ respectively. Since $(i,j)$ is a good pair, these events as well as $\mc{E}_S$ are independent, thus all the upscaling from the definitions of $X_i, X_j$ cancel out with their probabilities from the chain rule:
        \[
        \begin{aligned}
            \Cov \paren*{X_i, X_j \mid \mc{E}_S} &\leq \Ex \bracket*{X_iX_j\mid\mc{E}_S} - \Ex \bracket*{X_i\mid\mc{E}_S}\Ex \bracket*{X_j\mid\mc{E}_S} \\
            &= \Ex \bracket*{\Pest(u_i)\Pest(u_j)\mid\mc{E}_i \cap \mc{E}_j \cap \mc{E}_S} - \Ex \bracket*{\Pest(u_i)\mid\mc{E}_i \cap \mc{E}_S}\Ex \bracket*{\Pest(u_j)\mid\mc{E}_j \cap \mc{E}_S}.
        \end{aligned}
        \]
        
        This is close to $\Cov \paren*{\Pest(u_i)\Pest(u_j)\mid\mc{E}_i \cap \mc{E}_j \cap \mc{E}_S}$, however we must replace $\Ex \bracket*{\Pest(u_i)\mid\mc{E}_i \cap \mc{E}_S}$ with $\Ex \bracket*{\Pest(u_i)\mid\mc{E}_i \cap \mc{E}_j \cap \mc{E}_S}$ and $\Ex \bracket*{\Pest(u_j)\mid\mc{E}_j \cap \mc{E}_S}$ with $\Ex \bracket*{\Pest(u_j)\mid\mc{E}_i \cap \mc{E}_j \cap \mc{E}_S}$. By \Cref{eq:mean-shift-ih}, we know that both differences are at most $3\sigma_{a-1} \cdot f(|S|+\Tmax)$. We can then use the identity that $|ab-a'b'| \leq |a-a'|+|b-b'|$ which holds when $a,b,a',b' \in [0,1]$ to show that this difference is at most $12\sigma_{a-1} \cdot f(|S|)$. This gives us
        \begin{align}
            \Cov \paren*{X_i, X_j \mid \mc{E}_S} &\leq \Cov \paren*{\Pest(u_i)\Pest(u_j) \mid \mc{E}_i \cap \mc{E}_j \cap \mc{E}_S} + 12\sigma_{a-1} \cdot f(|S|) \notag \\
            &\leq \sqrt{\Var \paren*{\Pest(u_i) \mid \mc{E}_i \cap \mc{E}_j \cap \mc{E}_S}\Var \paren*{\Pest(u_j)\mid\mc{E}_i \cap \mc{E}_j \cap \mc{E}_S}} + 12\sigma_{a-1} \cdot f(|S|). \tag{\Cref{lem:cauchy-schwarz}}
        \end{align}
        
        We can now use the inductive hypothesis to show that $\Var \paren*{\Pest(u_i)\mid\mc{E}_i \cap \mc{E}_j \cap \mc{E}_S} \leq \sigma_{a-1}^2 \cdot g(|S|+\Tmax)$, and similarly for $u_j$. This gives us the final bound for our covariance:
        \[
        \Cov \paren*{X_i, X_j \mid \mc{E}_S} \leq \sigma_{a-1}^2 \cdot g(|S|+\Tmax) + 12\sigma_{a-1} \cdot f(|S|) \leq 13\sigma_{a-1} \cdot g(|S|).
        \]
        Thus, we can use this to bound the sum of covariances of all good pairs:
        \begin{equation}\label{eq:good-pairs}
            \mathop{\sum\sum}_{(i,j) \text{ good}} \Cov \paren*{X_i, X_j \mid \mc{E}_S} \leq \abs*{\Neighbors(v)}^2 \cdot 13\sigma_{a-1} \cdot g(|S|).
        \end{equation}
        
        Now to bound the bad pairs, we see that for any indicator variables $X_i,X_j$, $\Cov(X_i,X_j|A) \leq n^{2c\Tmax}$ for any event $A$, as indicator variables lie in the range $[0,n^{c\Tmax}]$. It suffices to simply bound the number of bad pairs which can exist and then multiply by the maximum possible covariance of each pair. Consider a pair $(i,j)$, this can be bad if at least one of three conditions hold: $\T(u_i) \cap \T(u_j) \neq \emptyset$, $\T(u_i) \cap S \neq \emptyset$, or $\T(u_j) \cap S \neq \emptyset$. From \Cref{lem:count-tree-containment}, we can see that any single vertex can be in at most $n^{q\cdot 2^{a-1}}$ trees of vertices in $\Neighbors(v)$. Thus, the total number of vertices with trees intersecting with $S$ is at most $|S| \cdot n^{q\cdot 2^{a-1}}$, which accounts for $2 \cdot \abs*{\Neighbors(v)} \cdot |S| \cdot n^{q\cdot 2^{a-1}}$ pairs. Similarly, the number of pairs vertices in $\Neighbors(v)$ which can intersect each other is at most $\abs*{\Neighbors(v)} \cdot \Tmax \cdot n^{q\cdot 2^{a-1}}$. This gives us the following bound for bad pairs:
        \[
        \begin{aligned}
            \mathop{\sum\sum}_{(i,j) \text{ bad}} \Cov \paren*{X_i, X_j \mid \mc{E}_S} &\leq n^{2c\Tmax} \cdot \paren*{2 \cdot \abs*{\Neighbors(v)} \cdot |S| \cdot n^{q\cdot 2^{a-1}} + \abs*{\Neighbors(v)} \cdot \Tmax \cdot n^{q\cdot 2^{a-1}}} \\
            &\leq \abs*{\Neighbors(v)} \cdot n^{q\cdot 2^{a-1}} \cdot n^{2c\Tmax} \cdot \paren*{2|S|+\Tmax}.
        \end{aligned}
        \] 
        Substituting this and \Cref{eq:good-pairs} into \Cref{eq:cov-split-good-bad} gives an overall bound on our variance in the high-degree case (noting that $\Neighbors(v) \geq d_{\B,1}^+(v) \geq n^{q\cdot2^a}/2k$ is a simple lower bound derivable from the conditioning on \Cref{claim:high-deg-bound-prob}).
        \[
        \begin{aligned}
            \Var(\Pest(v)\mid\mc{E}_S) &\leq \frac{1}{\alpha^2\abs*{\Neighbors(v)}^2} \paren*{\abs*{\Neighbors(v)}^2 \cdot 13\sigma_{a-1} \cdot g(|S|) + \abs*{\Neighbors(v)} \cdot n^{q\cdot 2^{a-1}} \cdot n^{2c\Tmax} \cdot \paren*{2|S|+\Tmax}} \\
            &\leq \frac{13\sigma_{a-1}\cdot g(|S|)}{\alpha^2} + \frac{2k \cdot n^{q\cdot2^{a-1}}\cdot n^{2c\Tmax}\cdot (2|S|+\Tmax)}{\alpha^2\cdot n^{q\cdot2^a}} \\
            &\leq \sigma_a^2/2\cdot g(|S|) + \frac{2k \cdot2n^{2c\Tmax}}{\alpha^2 n^q}\cdot \frac{|S|+\Tmax/2}{n^q} \\
            &\leq \sigma_a^2 \cdot g(|S|).
        \end{aligned}
        \]
        
        Now for the low-degree case, we can substitute the formula for $\Pest(v)$ as we are conditioning on $\mc{E}_v$
        \[
            \Var \paren*{\Pest(v) \mid \mc{E}_v \cap \mc{E}_S} = \Var \paren*{\clamp{\frac{\yout(v)+\zloest(v)-\zliest(v)}{\yin(v)+\yout(v)}}\ \middle|\ \mc{E}_v \cap \mc{E}_S}.
        \]
        
        Now, clamp only decreases variance (\Cref{lem:var-of-clamp}), so we can ignore it, then we can pull out constants and substitute the formulas for $\zloest$ and $\zliest$.
        \[
        \begin{aligned}
            \Var &\paren*{\Pest(v) \mid \mc{E}_v \cap \mc{E}_S} \\
            &\leq \frac{1}{(\yin(v)+\yout(v))^2}\Var \paren*{\frac{\abs*{\OL(v)}}{d} \cdot \displaystyle \sum_{(u,v) \in \Rsvsout(v)} \Pest(u) - \frac{\abs*{\IL(v)}}{d} \cdot \displaystyle \sum_{(v,u) \in \Rsvsin(v)} (1-\Pest(u))\ \middle|\ \mc{E}_v \cap \mc{E}_S} \\
            &= \frac{1}{(\yin(v)+\yout(v))^2}\Var \paren*{\displaystyle \sum_{(u,v) \in \Rsvsout(v)} \frac{\abs*{\OL(v)}}{d} \cdot \Pest(u) + \displaystyle \sum_{(v,u) \in \Rsvsin(v)} \frac{\abs*{\IL(v)}}{d} \cdot \Pest(u)\ \middle|\ \mc{E}_v \cap \mc{E}_S}.
        \end{aligned}
        \]
        
        The denominator of the term our front can be lower-bounded by $\alpha^2(\abs*{\IL(v)}+\abs*{\IL(v)})^2$. In addition, we can split the variance of a sum into a sum of pairs of covariances. For notational purposes, enumerate the neighbors of $v$ in $\Rsvsout(v)$ by $u^{out}_1, \dots, u^{out}_d$, and similarly for $\Rsvsin(v)$ by $u^{in}_1, \dots, u^{in}_d$.
        \[
        \begin{aligned}
            &\leq \frac{1}{\alpha^2\paren*{\abs*{\IL(v)}+\abs*{\OL(v)}}^2} \cdot \paren*{\begin{array}{rl}
                &\displaystyle \sum_{i=1}^d \sum_{j=1}^d \Cov \paren*{\frac{\abs*{\OL(v)}}{d}\Pest(u^{out}_i), \frac{\abs*{\OL(v)}}{d}\Pest(u^{out}_j) \ \middle|\  \mc{E}_v \cap \mc{E}_S} \\
                + 2 \cdot&\displaystyle \sum_{i=1}^d \sum_{j=1}^d \Cov \paren*{\frac{\abs*{\OL(v)}}{d}\Pest(u^{out}_i), \frac{\abs*{\IL(v)}}{d}\Pest(u^{in}_j) \ \middle|\  \mc{E}_v \cap \mc{E}_S} \\
                + &\displaystyle \sum_{i=1}^d \sum_{j=1}^d \Cov \paren*{\frac{\abs*{\IL(v)}}{d}\Pest(u^{in}_i), \frac{\abs*{\IL(v)}}{d}\Pest(u^{in}_j) \ \middle|\  \mc{E}_v \cap \mc{E}_S}\\
            \end{array}} \\
            &= \frac{1}{\alpha^2\paren*{\abs*{\IL(v)}+\abs*{\OL(v)}}^2} \cdot \paren*{\begin{array}{rl}
                \frac{\abs*{\OL(v)}^2}{d^2}\cdot&\displaystyle \sum_{i=1}^d \sum_{j=1}^d \Cov \paren*{\Pest(u^{out}_i), \Pest(u^{out}_j) \ \middle|\  \mc{E}_v \cap \mc{E}_S} \\
                + \frac{2\abs*{\OL(v)}\abs*{\IL(v)}}{d^2} \cdot&\displaystyle \sum_{i=1}^d \sum_{j=1}^d \Cov \paren*{\Pest(u^{out}_i), \Pest(u^{in}_j) \ \middle|\  \mc{E}_v \cap \mc{E}_S} \\
                + \frac{\abs*{\IL(v)}^2}{d^2}\cdot&\displaystyle \sum_{i=1}^d \sum_{j=1}^d \Cov \paren*{\Pest(u^{in}_i), \Pest(u^{in}_j) \ \middle|\  \mc{E}_v \cap \mc{E}_S}.\\
            \end{array}} \\
        \end{aligned}
        \]
        
        Now consider any of these covariance terms between vertices which we will name $u'_1$ and $u'_2$. By Cauchy-Schwarz, we can split this into a geometric mean of variances. $\Cov\paren*{\Pest(u'_1), \Pest(u'_2)\mid \mc{E}_v \cap \mc{E}_S} \leq \sqrt{\Var\paren*{\Pest(u'_1)\mid \mc{E}_v \cap \mc{E}_S}\Var\paren*{\Pest(u'_2)\mid \mc{E}_v \cap \mc{E}_S}}$. Similarly to before, for the first variance, we can split up $\mc{E}_v$ into $\mc{E}_{u'_1} \cap \mc{E}_{v-{u'_1}}$ where $\mc{E}_{u'_1}$ corresponds to the event where $\T(u'_1) \subseteq W$ and $\mc{E}_{v-{u'_1}}$ corresponds to the event where $\T(v) - \T(u'_1) \subseteq W$. This holds because any tree of a child of a low-degree vertex must be a subset of the tree of its parent. We then group the second event with $\mc{E}_S$ and use the inductive hypothesis to show that $\Var\paren*{\Pest(u'_1)\mid \mc{E}_v \cap \mc{E}_S} \leq \sigma_{a-1}^2 \cdot g(|S|+\Tmax)$. The same logic applies to $u'_2$, so we know that this bound also applies to each covariance term. This simplifies the expression above to the following  
        \[
        \begin{aligned}
            &\leq \frac{1}{\alpha^2\paren*{\abs*{\IL(v)}+\abs*{\OL(v)}}^2} \cdot \paren*{\frac{\paren*{\abs*{\IL(v)}+\abs*{\OL(v)}}^2}{d^2} \cdot d^2 \cdot \sigma_{a-1}^2 \cdot g(|S|+\Tmax)} \\
            &= \frac{1}{\alpha^2} \cdot \sigma_{a-1}^2 \cdot g(|S|+\Tmax) \leq \sigma_a^2 \cdot g(|S|),
        \end{aligned}
        \]
        which concludes the proof of \Cref{claim:conditional-variance}.
    \end{proof}
    
    Next, we prove the second claim of \Cref{lem:variance+mean-shift} which uses the above.
    
    \begin{claim}\label{claim:mean-shift}
        Take a vertex $v$, and vertex sets $R$ and $S$, such that $T(v)$, $R$, and $S$ are pairwise disjoint. Let $\mc{E}_v$, $\mc{E}_R$, and $\mc{E}_S$ be the events that $T(v) \subseteq W$, $R \subseteq W$, and $S \subseteq W$ respectively. It holds that
        \[
            \card{\Ex[\Pest(v) \mid \mc{E}_v \cap \mc{E}_R]- \Ex[\Pest(v) \mid \mc{E}_v \cap \mc{E}_R \cap \mc{E}_S]} \leq 3\sigma_{a-1} \cdot f(\card{R}+ \card{S}).
        \]
    \end{claim}
    
    \begin{proof}
        First, we consider the case where $v$ is a high-degree vertex.
        Recall the formulation of $\Pest(v)$ for high-degree vertices:
        $$
        \Pest(v) = \altclamp \left(
            \begin{cases}
                \frac{\alpha+1}{\alpha} \cdot \frac{\outdegB(v)}{\degB(v)} - \frac{1}{\alpha} \cdot \zbest & \text{if } \chi(v) = k \\
                \frac{k-1}{k-\chi(v)} \cdot \frac{\outdegB(v)}{\degB(v)} - \frac{\chi(v)-1}{k-\chi(v)} \cdot \zbest & \text{if } \chi(v) \neq k  
            \end{cases}
            \right).
        $$
        Note that after fixing the randomness of $\B$, both $\degB(v)$ and $\outdegB(v)$ are constants.
        Also for any color $\chi(v)$, it holds $\frac{\chi(v) - 1}{k - \chi(v)} \leq \frac{1}{\alpha}$, and we can remove the $\altclamp$ with \cref{lem:conditional-clamping}.
        Therefore, it follows:
        \begin{align*}
           \mathsf{(LHS)}
            \leq \frac{1}{\alpha} \card{\Ex\left[\zbest(v) \mid \mc{E}_v \cap \mc{E}_R\right]
            - \Ex\left[\zbest(v) \mid \mc{E}_v \cap \mc{E}_R \cap \mc{E}_S\right]
            } &+  \sqrt{\Var\paren*{\Pest(v) \mid \mc{E}_v \cap \mc{E}_R}} \\
            &+ \sqrt{\Var\paren*{\Pest(v) \mid \mc{E}_v \cap \mc{E}_R \cap \mc{E}_S}}.
        \end{align*}

        By \Cref{claim:conditional-variance}, we can bound the last two terms by $\sqrt{\sigma_a^2 \cdot g(|R|)}$ and $\sqrt{\sigma_a^2 \cdot g(|R|+|S|)}$ for a total of at most $2\sigma_a \cdot f(|R|+|S|)$. Thus, it suffices to show that
        \[
        \frac{1}{\alpha} \card{\Ex\left[\zbest(v) \mid \mc{E}_v \cap \mc{E}_R\right]
            - \Ex\left[\zbest(v) \mid \mc{E}_v \cap \mc{E}_R \cap \mc{E}_S\right]
            } \leq \sigma_a \cdot f(|R|+|S|).
        \]
        Let this term be denoted by $(*)$, we can expand the definition of $\zbest(v)$, and let $X_i$ denote the indicator variable for $u_i \in \Neighbors(v)$, the $i$-th sampled neighbor of $v$ in the Horvitz-Thompson average. 
        \[
        (*) \leq \frac{1}{\alpha \card{\Neighbors(v)}} \sum_{i=1}^{\card{\Neighbors(v)}} \card{\Ex \bracket*{X_i \mid \mc{E}_v \cap \mc{E}_R} - \Ex \bracket*{X_i \mid \mc{E}_v \cap \mc{E}_R \cap \mc{E}_S}}.
        \]
    
        Now, we bound the sum by dividing its terms into two groups: (1) terms corresponding to neighbors such that $\T(u_i)$ is disjoint from $R$ and $S$ (note that all high-degree neighbors are in this group).
        In this case, by the induction hypothesis, we can bound each term by $3\sigma_{a-1} \cdot f(\card{R} + \card{S})$. This is because if $\T(u_i)$ is disjoint from $R$ and $S$, then the event where $\T(u_i) \subseteq W$ is independent from $\mc{E}_R \cap \mc{E}_S$, meaning that the probability of success cancels with the upscaling of $X_i$.
        Since there are a total of $\card{\Neighbors(v)}$ terms, the total contribution of these terms to the sum is at most $\card{\Neighbors(v)} \cdot 3\sigma_{a-1} \cdot f(\card{R} + \card{S})$.
        (2) terms corresponding to neighbors with $T(u_i)$ intersecting either $R$ or $S$.
        In this case, we bound each term trivially by the maximum value of an indicator variable, which is $n^{c \Tmax}$.
        By \Cref{lem:count-tree-containment}, each vertex in $R$ or $S$ is included in the tree of at most $n^{q2^{a-1}}$ neighbors of $v$. 
        Consequently, the number of neighbors $u_i$ for which $T(u_i)$ intersects with $R$ or $S$ is at most $(\card{R} + \card{S}) n^{q2^{a-1}}$, 
        and the total contribution of such terms to the sum can be bounded by 
        $(\card{R} + \card{S}) n^{q2^{a-1}} \cdot n^{c \Tmax}$.
    
        Overall, we can bound the sum as follows (noting that $\Neighbors(v) \geq d_{\B,1}^+(v) \geq n^{q\cdot2^a}/2k$ is a simple lower-bound derivable from the conditioning on \Cref{claim:high-deg-bound-prob}):
        \begin{align}
            (*)
            &\leq \frac{1}{\alpha \card{\Neighbors(v)}} \sum_{i=1}^{\card{\Neighbors(v)}} \card{\Ex \bracket*{X_i \mid \mc{E}_v \cap \mc{E}_R} - \Ex \bracket*{X_i \mid \mc{E}_v \cap \mc{E}_R \cap \mc{E}_S}} \notag \\
            &\leq \frac{1}{\alpha \card{\Neighbors(v)}}\left(\card{\Neighbors(v)} \cdot 3\sigma_{a-1} \cdot f(\card{R} + \card{S}) + (\card{R} + \card{S}) n^{q2^{a-1}} \cdot n^{c \Tmax}\right) \notag \\
            &\leq \frac{3\sigma_{a-1}}{\alpha} f(\card{R} + \card{S}) + 
            (\card{R} + \card{S}) \frac{2k \cdot n^{q2^{a-1}} \cdot n^{c \Tmax}}{\alpha \cdot n^{q2^{a}}} \label{eq:conditional-expectation-1} \\
            &\leq \sigma_{a} \cdot f(\card{R} + \card{S}). \label{eq:conditional-expectation-2}
        \end{align}
        Here, \eqref{eq:conditional-expectation-1} follows from the fact that $v$ is high-degree (i.e., $\card{\Neighbors(v)} \geq n^{q2^{a}}$), and \eqref{eq:conditional-expectation-2} holds since $3\sigma_{a - 1} / \alpha \leq \sigma_{a} / 2$ and 
        $$
        (\card{R} + \card{S}) \frac{n^{q2^{a-1}} \cdot n^{c \Tmax}}{\alpha \cdot n^{q2^{a}}}
        \leq \frac{\card{R} + \card{S}}{n^q} \cdot \frac{n^{c\Tmax}}{\alpha \cdot n^q} 
        \leq f(\card{R} + \card{S}) \frac{\sigma_{a}}{2}.
        $$
        This concludes the proof for high-degree vertices.
    
        Next, we move on to proving the claim for a low-degree vertex $v$.
        Recall the definition of $\Pest(v)$ for low-degree vertices:
        $$
        \Pest(v) = \clamp{\frac{\yout(v) + \zloest(v) - \zliest(v)}{\yin(v)+\yout(v)}}.
        $$
        We can remove the clamp (by \cref{lem:conditional-clamping}) and the constants $\yout(v)$ and $\yin(v)$ and lower-bound the denominator by $\alpha (\card{\IL(v)} + \card{\OL(v)})$ to get (noting that if $\IL(v)$ and $\OL(v)$ are empty then the case is trivial):

        \[
        \begin{aligned}
            \mathsf{(LHS)} &\leq \frac{\card{\Ex[\zloest(v) - \zliest(v) \mid \mc{E}_v \cap \mc{E}_R] -\Ex[\zloest(v) - \zliest(v) \mid \mc{E}_v \cap \mc{E}_R \cap \mc{E}_S]}}{\alpha (\card{\IL(v)} + \card{\OL(v)})} \\
            &+ \sqrt{\Var\paren*{\Pest(v) \mid \mc{E}_v \cap \mc{E}_R}}
            + \sqrt{\Var\paren*{\Pest(v) \mid \mc{E}_v \cap \mc{E}_R \cap \mc{E}_S}}.
        \end{aligned}
        \]
        Again, by \Cref{claim:conditional-variance}, we can bound the right two terms by $\sqrt{\sigma_a^2 \cdot g(|R|)}$ and $\sqrt{\sigma_a^2 \cdot g(|R|+|S|)}$ for a total of at most $2\sigma_a \cdot f(|R|+|S|)$. Thus, it suffices to show that
        \begin{equation}\label{eq:conditional-expectation-3}
            \frac{\card{\Ex[\zloest(v) - \zliest(v) \mid \mc{E}_v \cap \mc{E}_R] -\Ex[\zloest(v) - \zliest(v) \mid \mc{E}_v \cap \mc{E}_R \cap \mc{E}_S]}}{\alpha (\card{\IL(v)} + \card{\OL(v)})} \leq \sigma_a \cdot f(|R|+|S|).
        \end{equation}
        
        Let us restrict our attention to this difference for $\zliest(v)$,
        letting $(*)$ denote the difference $\card{\Ex[\zliest(v) \mid \mc{E}_v \cap \mc{E}_R] -\Ex[\zliest(v) \mid \mc{E}_v \cap \mc{E}_R \cap \mc{E}_S]}$.
        If $\IL(v)$ is empty, the difference is 0. Otherwise, recall the definition of $\zliest$:
        $$
        \zliest(v) = \frac{\card{\IL(v)}}{d} \sum_{i = 1}^d \Pest(u_i),
        $$
        where $u_1, u_2, \ldots, u_d$ are the sampled neighbors $\Rsvsin(v)$. From this, we obtain:
        \begin{equation}
        (*) \leq \frac{\card{\IL(V)}}{d} \sum_{i=1}^d \card{\Ex[\Pest(u_i) \mid \mc{E}_v \cap \mc{E}_R] - \Ex[\Pest(u_i) \mid \mc{E}_v \cap \mc{E}_R \cap \mc{E}_S]}.
         \label{eq:conditional-expectation-4}
        \end{equation}
        
        Now, we bound each term using the induction hypothesis.
        For a neighbor $u_i$, let $\mc{E}_i$ denote the event that $T(u_i) \subseteq W$, and $\mc{E}_{v - i}$ denote the event that $T(v) \setminus T(u_i) \subseteq W$.
        We can rewrite the event $\mc{E}_v$ as $\mc{E}_i \cap \mc{E}_{v-i}$, and it follows:
        \begin{align}
            |\Ex[\Pest(u_i) \mid \mc{E}_v \cap \mc{E}_R] - \Ex[\Pest(u_i) &\mid \mc{E}_v \cap \mc{E}_R \cap \mc{E}_S]| \notag \\
            &= \card{\Ex[\Pest(u_i) \mid \mc{E}_i \cap (\mc{E}_{v-i} \cap \mc{E}_R)] - \Ex[\Pest(u_i) \mid \mc{E}_i \cap (\mc{E}_{v-i} \cap \mc{E}_R) \cap \mc{E}_S]}  \notag \\
            &\leq 3\sigma_{a - 1} \cdot f(\card{R \cup (T(v) \setminus T(u_i))} + \card{S})  \label{eq:conditional-expectation-5} \\
            &\leq 3\sigma_{a - 1} \cdot f(\card{R} + \card{S} + \card{\Tmax}) \notag \\
            &\leq 3\sigma_{a - 1} \cdot 2f(\card{R} + \card{S}), \notag
        \end{align}
        where \eqref{eq:conditional-expectation-5} can be derived by invoking the induction hypothesis for $u_i$, $R' = R \cup (T(v) \setminus T(u_i)$, and $S$.
        Observe that since $T(v)$, $R$, and $S$ are pairwise disjoint, so are $T(u_i)$, $R'$, and $S$ (as the new sets are obtained by moving some vertices from the first set to the second).
    
        Plugging the upper bound for each term into the sum in \eqref{eq:conditional-expectation-4} yields
        $$
        \card{\Ex[\zliest(v) \mid \mc{E}_v \cap \mc{E}_R] -\Ex[\zliest(v) \mid \mc{E}_v \cap \mc{E}_R \cap \mc{E}_S]} \leq \card{\IL(v)} \cdot 6\sigma_{a - 1} \cdot f(\card{R} + \card{S}).
        $$
        Similarly, it can be shown for the $\zloest$ that
        $$
        \card{\Ex[\zloest(v) \mid \mc{E}_v \cap \mc{E}_R] -\Ex[\zloest(v) \mid \mc{E}_v \cap \mc{E}_R \cap \mc{E}_S]} \leq \card{\OL(v)} \cdot 6\sigma_{a - 1} \cdot f(\card{R} + \card{S}).
        $$
        Putting everything together, we see that the left hand side of \eqref{eq:conditional-expectation-3} is at most
        $$
        \frac{6\sigma_{a-1} \card{\IL(v)} + 6\sigma_{a-1} \card{\OL(v)}}{\alpha(\card{\IL(v)} + \card{\OL(v)})} \cdot f(\card{R} + \card{S}) \leq \sigma_{a} \cdot f(\card{R} + \card{S}).
        $$
        This concludes the proof for the low-degree vertices, and in turn, the claim.
    \end{proof}
Thus, the combination of \Cref{claim:conditional-variance} and \Cref{claim:mean-shift} conclude the proof of \Cref{lem:variance+mean-shift}
\end{proof}

Now, we show that our expected $\Pest(\cdot)$ is close to $\pos(\cdot)$.

\begin{lemma}\label{lem:Ex-P-pesth}
    For any vertex $v$, with color $\chi(v) =a$, let $\mc{E}_v$ be the event where $\T(v) \subseteq W$. It holds that
    \[
    \abs*{\Ex \bracket*{\Pest(v) \mid \mc{E}_v} - \Pesth(v)} \leq \delta_a.
    \]
\end{lemma}

\begin{proof}  
    We prove this lemma inductively, let $\chi(v) = 1$. The low degree case is trivial as $\zliesth(v) = \zloesth(v) = \zliest(v) = \zloest(v) = 0$, so $\Pesth(v) = \Pest(v)$ with probability 1. If $v$ is high-degree, then $\zbesth = \zbest = 0$, so again we have $\Pesth(v) = \Pest(v)$ concluding the base case. Now consider $\chi(v) = a$. We can assume the following inductive hypothesis for any vertex $u$ with $\chi(u) < a$ and $\mc{E}_u$ being the event where $\T(u) \subseteq W$.
    \begin{equation}\label{eq:Ex-P-pesth-ih}
        \abs*{\Ex \bracket*{\Pest(u) \mid \mc{E}_u} - \Pesth(u)} \leq \delta_{a-1}.
    \end{equation}

    Again, we begin with the case where $v$ is low-degree and simply expand the formulas.
    \[
    \abs*{\Ex \bracket*{\Pest(v)\mid\mc{E}_v} - \Pesth(v)} = \abs*{\Ex \bracket*{\clamp{\frac{\yout(v)- \zloest(v)+\zliest(v)}{\yin(v) + \yout(v)}}\ \middle|\ \mc{E}_v}- \clamp{\frac{\yout(v)- \zloesth(v)+\zliesth(v)}{\yin(v) + \yout(v)}}}
    \]
    
    Now, we can apply \Cref{lem:conditional-clamping} to get the following.
    \begin{align}
        &\abs*{\Ex \bracket*{\Pest(v)\mid\mc{E}_v} - \Pesth(v)} \leq \abs*{\Ex \bracket*{\frac{\yout(v)- \zloest(v)+\zliest(v)}{\yin(v) + \yout(v)} \ \middle|\  \mc{E}_v}- \frac{\yout(v)- \zloesth(v)+\zliesth(v)}{\yin(v) + \yout(v)}} + \sqrt{\Var(\Pest(v) \mid \mc{E}_v)} \notag \\
        &\leq \frac{1}{\alpha \cdot (\abs*{\IL(v)} + \abs*{\OL(v)})} \cdot \paren*{\abs*{\Ex \bracket*{\zloest(v) \ \middle|\  \mc{E}_v} - \zloesth(v)} + \abs*{\Ex \bracket*{\zliest(v) \ \middle|\  \mc{E}_v} - \zliesth(v)}} + \sqrt{\Var(\Pest(v) \mid \mc{E}_v)}. \label{eq:Pest-mean-eq-1}
    \end{align}
    Here, the second step is just simplification and applying \Cref{eq:yinout-simple}, noting that if $\IL(v)$ and $\OL(v)$ are empty the lemma is trivial. Now, we bound $\abs*{\Ex \bracket*{\zliest(v) \mid \mc{E}_v} - \zliesth(v)}$. An analogous bound applies to $\abs*{\Ex \bracket*{\zloest(v) \mid \mc{E}_v} - \zloest(v)}$. To do this, first note that if $\IL(v)$ is empty, then the difference is 0. We can expand the definition of both $\zliest(v)$ and $\zliesth(v)$. For notation, let the $d$ neighbors of $v$ in $\Rsvsin(v)$ be $u^{in}_1, \dots, u^{in}_d$.
    \[
    \begin{aligned}
        \abs*{\Ex \bracket*{\zliest(v) \ \middle|\  \mc{E}_v} - \zliesth(v)} &= \abs*{\Ex \bracket*{\frac{\abs*{\IL(v)}}{d} \displaystyle \sum_{i=1}^d \Pest(u^{in}_i) \ \middle|\  \mc{E}_v} - \frac{\abs*{\IL(v)}}{d} \displaystyle \sum_{i=1}^d \Pesth(u^{in}_i)} \\
        &\leq \frac{\abs*{\IL(v)}}{d} \sum_{i=1}^d \abs*{\Ex \bracket*{\Pest(u^{in}_i) \ \middle|\  \mc{E}_v} - \Pesth(u^{in}_i)}.
    \end{aligned}
    \]

    Now, from the second part of \Cref{lem:variance+mean-shift}, we know that $\abs*{\Ex \bracket*{\Pest(u^{in}_i) \mid \mc{E}_v} - \Ex \bracket*{\Pest(u^{in}_i) \mid \mc{E}_{u^{in}_i}}} \leq \delta_{a-1}$ (Where $\mc{E}_{u^{in}_i}$ is the event where $\T(u^{in}_i) \subseteq W$). This is because $\mc{E}_v = \mc{E}_{u^{in}_i} \cap \mc{E}_v$, so we can apply our lemma and get a bound of $3\sigma_{a-1}\cdot f(\Tmax) \leq \delta_{a-1}$.
    This fact along with \Cref{eq:Ex-P-pesth-ih} gives us 
    \[
    \abs*{\Ex \bracket*{\zliest(v) \ \middle|\  \mc{E}_v} - \zliesth(v)} \leq \frac{\abs*{\IL(v)}}{d} \sum_{i=1}^d \paren*{\delta_{a-1} + \delta_{a-1}} \leq \abs*{\IL(v)}\cdot 2\delta_{a-1}.
    \]

    Analogously, for $\zloest(v)$ we have
    \[
    \abs*{\Ex \bracket*{\zloest(v) \ \middle|\  \mc{E}_v} - \zloesth(v)} \leq \abs*{\OL(v)}\cdot 2\delta_{a-1}.
    \]
    Plugging these back into \Cref{eq:Pest-mean-eq-1} along with the first part of \Cref{lem:variance+mean-shift} gives
    \[
    \begin{aligned}
        \abs*{\Ex \bracket*{\Pest(v)\mid\mc{E}_v} - \Pesth(v)} &\leq \frac{1}{\alpha \cdot (\abs*{\IL(v)} + \abs*{\OL(v)})} \cdot \paren*{2\delta_{a-1}\cdot\abs*{\OL(v)} + 2\delta_{a-1}\cdot\abs*{\IL(v)}} + \sigma_{a} \\
        &\leq \frac{2\delta_{a-1}}{\alpha}+\sigma_{a} \leq \delta_a.
    \end{aligned}
    \]
    This concludes the proof for the low-degree case. 
    
    Now, we assume that $v$ is high-degree. We can ignore the conditioning on $\mc{E}_v$ as it happens with probability 1. We bound the desired difference by substituting their formulas and using \Cref{lem:conditional-clamping} as we did in the low-degree case. As done previously, we will use the case where $\chi(v) = k$ for the formula as it is an upper-bound on both cases because $\frac{\chi(v) - 1}{k - \chi(v)} \leq k - 2 \leq \frac{1}{\alpha}$.
    \begin{align}
        \abs*{\Ex \bracket*{\Pest(v)} - \Pesth(v)} &\leq \abs*{\Ex \bracket*{\frac{\alpha+1}{\alpha} \cdot\frac{\outdegB(v)}{\degB(v)} - \frac{\zbest(v)}{\alpha}} - \paren*{\frac{\alpha+1}{\alpha} \cdot\frac{\outdegB(v)}{\degB(v)} - \frac{\zbesth(v)}{\alpha}}} + \sqrt{\Var(\Pest(v))} \notag\\
        &\leq \frac{1}{\alpha}\abs*{\Ex[\zbest(v)]-\zbesth(v)} + \sqrt{\Var(\Pest(v))} \label{eq:Pest-mean-eq-hi}.
    \end{align}

    We see that we can upper-bound the second term with \Cref{lem:variance+mean-shift}, so we focus on the first. We see that for $u\in \Neighbors(v)$ with indicator variable $X_u$ and event $\mc{E}_u$ corresponding to $\T(u) \subseteq W$, we know that $\Ex[X_u] = \Ex[\Pest(u)\mid\mc{E}_u]$. Thus, we can bound the desired term as follows using triangle inequality and \Cref{eq:Ex-P-pesth-ih}.
    \[
    \begin{aligned}
        \abs*{\Ex[\zbest(v)]-\zbesth(v)} &\leq \abs*{\frac{1}{\abs*{\Neighbors(v)}} \sum_{u \in \Neighbors(v)} X_u -\frac{1}{\abs*{\Neighbors(v)}} \sum_{u \in \Neighbors(v)} \Pesth(u)} \\
        &\leq \frac{1}{\abs*{\Neighbors(v)}} \sum_{u \in \Neighbors(v)} \abs*{\Ex[\Pest(u)\mid\mc{E}_u] - \Pesth(u)} \leq \frac{1}{\abs*{\Neighbors(v)}} \sum_{u \in \Neighbors(v)} \delta_{a-1} \leq \delta_{a-1}.
    \end{aligned}
    \]
    Substituting this back into \Cref{eq:Pest-mean-eq-hi} along with the first part of \Cref{lem:variance+mean-shift} gives
    \[
    \abs*{\Ex \bracket*{\Pest(v)} - \Pesth(v)} \leq \frac{1}{\alpha}\cdot \delta_{a-1} + \sigma_a \leq \delta_a.
    \]
    This concludes the proof of the high-degree case and the lemma.
\end{proof}

Since our expectation is close, we also prove that our estimate $\Pest(\cdot)$ is close to $\Pesth(\cdot)$ with good probability, even when conditioned on an arbitrary other vertex's success. This will become relevant when extending this result about vertices to a result about edges.

\begin{lemma}\label{lem:Pr-P-pesth}
    For any vertices $u,v$,  let $\mc{E}_u,\mc{E}_v$ be the events where $\T(u),\T(v) \subseteq W$. It holds that
    \[
    \Pr \paren*{\abs*{\Pest(v) - \Pesth(v)} > 3\delta_k \mid \mc{E}_v, \mc{E}_u} \leq \delta_k.
    \]
\end{lemma}

\begin{proof}
    The proof comes from a combination of \Cref{lem:Ex-P-pesth} and \Cref{lem:chebyshev}. From Chebyshev, we know that 
    \[
    \Pr \paren*{\abs*{\Pest(v) - \Ex \bracket*{\Pest(v) \mid \mc{E}_v, \mc{E}_u}} > \delta_k \mid \mc{E}_v, \mc{E}_u} \leq \frac{\Var \paren*{\Pest(v) \mid \mc{E}_v, \mc{E}_u}}{\delta_k^2}.
    \]

    We can now use the first part of \Cref{lem:variance+mean-shift}  to see that $\Var \paren*{\Pest(v) \mid \mc{E}_v, \mc{E}_u} \leq \sigma_k^2\cdot g(|\Tmax|) \leq 2\sigma_k^2$. Thus, our probability is at most $2\sigma_k^2/\delta_k^2 \leq \delta_k$. Thus, it suffices to show that $\abs*{\Ex \bracket*{\Pest(v) \mid \mc{E}_v, \mc{E}_u} - \Pesth(v)} \leq 2\delta_k$. This can be seen with a triangle inequality, \Cref{lem:Ex-P-pesth}, and an application of the second part of \Cref{lem:variance+mean-shift} which shows that $\abs*{\Ex \bracket*{\Pest(v) \mid \mc{E}_v, \mc{E}_u}-\Ex \bracket*{\Pest(v) \mid \mc{E}_v}} \leq 3\sigma_k \cdot f(|\Tmax|) \leq 6\sigma_k \leq \delta_k$. Combining these gives error of $\delta_k+\delta_k = 2\delta_k$, proving the claim.
\end{proof}

Now, we extend this result to edges, making use of the extra conditioning that we added in the previous lemma.

\begin{lemma}\label{lem:P-pesth-prob-edge}
    For any edge $(u,v)$, it holds that
    \[
    \Pr \paren*{\abs*{\Pest(u)\cdot(1-\Pest(v))-\Pesth(u)\cdot(1-\Pesth(v))} > 9\delta_k \mid \mc{E}_u, \mc{E}_v} \leq 2\delta_k.
    \]
\end{lemma}
\begin{proof}
    The proof is almost identical to \Cref{lem:pesth-pos-prob-edge}. Condition on events $\mc{E}_u, \mc{E}_v$. By a union bound on \Cref{lem:Pr-P-pesth} for vertices $u$ and $v$, it suffices to show that if we condition further on the events where $\abs*{\Pest(v)-\Pesth(v)} \leq 3\delta_k$ and $\abs*{\Pest(u)-\Pesth(u)} \leq 3\delta_k$, then with probability 1 we have
    \[
    \abs*{\Pest(u)\cdot(1-\Pest(v))-\Pesth(u)\cdot(1-\Pesth(v))} \leq 9\delta_k.
    \]
    We show this as follows:
    \begin{align*}
        \abs*{\Pest(u)\cdot(1-\Pest(v))-\Pesth(u)\cdot(1-\Pesth(v))} &\leq \abs*{\Pest(u)-\Pesth(u)} + \abs*{\Pest(u)\Pest(v)-\Pesth(u)\Pesth(v)} \\
        &\leq 3\delta_k + \abs*{\Pest(u)\Pest(v)-\Pest(u)\Pesth(v)} + \abs*{\Pest(u)\Pesth(v)-\Pesth(u)\Pesth(v)} \\
        &= 3\delta_k + \Pest(u)\cdot\abs*{\Pest(v)-\Pesth(v)} + \Pesth(v)\cdot\abs*{\Pest(u)-\Pesth(u)}\\
        &\leq 3\delta_k + 3\delta_k + 3\delta_k = 9\delta_k. \qedhere
    \end{align*}
\end{proof}

We have now succeeded in bounding the probability that an edge's $\Pest(\cdot)$ value is close to its $\Pesth(\cdot)$ value. We now extend this to a result bounding the difference between $\cutval$ (the value returned by the algorithm) and $\val{G'}{\Pesth}$. We do this by bounding the both variance of $\cutval$ and the difference between its mean and $\val{G'}{\Pesth}$.

\begin{lemma}\label{lem:mean-cutval} It holds that
    \[
    \abs*{\Ex \bracket*{\cutval} - \val{G'}{\Pesth}} \leq \delta_{k+1}.
    \]
\end{lemma}

\begin{proof}
    Let $X_e$ be the standard indicator variable for an edge $e$. We create another variable $Y_e$ which takes the value $X_e$ if $e \in \C$, and 0 otherwise. Thus, we can see that $\cutval = \sum_{e \in E} Y_e$. In addition, for each edge $\Pr(e \in \C) = \frac{|\C|}{|E|}$. Thus $\Ex[Y_e] = \frac{|\C|}{|\E|} \Ex[X_e]$. We use this fact to expand the desired difference in the lemma statement and bound it, making use of the fact that $\abs*{\Ex\bracket*{\Pest(u)\cdot(1-\Pest(v)) \mid \mc{E}_u, \mc{E}_v}-\Pesth(u)\cdot(1-\Pesth(v))} \leq 18\delta_k$ which can be seen from \Cref{lem:P-pesth-prob-edge} and \Cref{lem:custom-prob-ex}.
    \begin{align*}
        \abs*{\Ex \bracket*{\cutval} - \val{G'}{\Pesth}} &= \abs*{\Ex \bracket*{\frac{1}{|\C|}\sum_{e \in E} Y_e} - \frac{1}{|E|}\sum_{e \in E} \Pesth(e)} \\
        &\leq \frac{1}{|E|}\sum_{e \in E} \abs*{\frac{|E|}{|\C|}\Ex[Y_e]-\Pesth(e)}\\
        &= \frac{1}{|E|}\sum_{e \in E} \abs*{\Ex[X_e]-\Pesth(e)} \\
        &= \frac{1}{|E|}\sum_{(u,v) \in E} \abs*{\Ex\bracket*{\Pest(u)\cdot(1-\Pest(v)) \mid \mc{E}_u, \mc{E}_v}-\Pesth(u)\cdot(1-\Pesth(v))} \\
        &\leq \frac{1}{|E|}\sum_{(u,v) \in E} 18\delta_k \leq 18\delta_k \leq \delta_{k+1}. \qedhere
    \end{align*}
    
\end{proof}

\begin{lemma}\label{lem:var-cutval}
It holds that
    \[
    \Var \paren*{\cutval} \leq \sigma_{k+1}^2.
    \]
\end{lemma}

\begin{proof}
    For an edge $e = (u,v)$, let $X_e$ be its standard indicator variable and let $\Pest(e) = \Pest(u) \cdot (1-\Pest(v))$. We can expand the formula for $\cutval$ \Cref{line:cutval-def} as follows:
    \begin{equation}\label{eq:var-cutval-start}
    \begin{aligned}
        \Var(\cutval) &= \Var \paren*{\frac{1}{|\C|} \sum_{e\in C} X_e} \\
        &\leq \frac{1}{|\C|^2} \sum_{e_1 \in C} \sum_{e_2 \in C} \Cov \paren*{X_{e_1}, X_{e_2}}.
    \end{aligned}
    \end{equation}
    
    Similar to before, we split the sum of covariances into good pairs and bad pairs. Call a pair of edges $e_1 = (u_1,v_1),e_2=(u_2,v_2)$ good if $\T(u_1) \cup \T(v_1)$ is disjoint from $\T(u_2)\cup\T(v_2)$. Let $\mc{E}_1$ be the event where $\T(u_1) \cup \T(v_1) \subseteq W$, and similarly define $\mc{E}_2$. Then, if $e_1,e_2$ is a good pair, then the events $\mc{E}_1, \mc{E}_2$ are independent. This allows us to say that $\Ex[X_{e_1}X_{e_2}] = \Ex[\Pest(e_1)\Pest(e_2)\mid\mc{E}_1, \mc{E}_2]$. The logic is the same as for vertices, we use chain rule over both events $\mc{E}_1, \mc{E}_2$, and the independence implies that the probability of both occurring scales inversely proportional to the upscaling of our Horvitz-Thompson average. 
    \begin{equation}\label{eq:var-cutval-1}
    \begin{aligned}
        \Cov \paren*{X_{e_1}, X_{e_2}} &= \Ex \bracket*{X_{e_1}X_{e_2}} - \Ex \bracket*{X_{e_1}}\Ex \bracket*{X_{e_2}} \\
        &= \Ex \bracket*{\Pest(e_1)\Pest(e_2) \mid \mc{E}_1, \mc{E}_2} - \Ex \bracket*{\Pest(e_1)\mid \mc{E}_1}\Ex \bracket*{\Pest(e_2)\mid \mc{E}_2}.
    \end{aligned}
    \end{equation}

    We see that this is close to $\Cov \paren*{\Pest(e_1),\Pest(e_2) \mid \mc{E}_1, \mc{E}_2}$, the only change that needs to be made is the conditioning in the terms on the right. Because of this, we bound the shift in mean of each $\Pest(\cdot)$. We focus on the first as the second is analogous.
    \[
    \begin{aligned}
        \abs*{\Ex \bracket*{\Pest(e_1)\mid \mc{E}_1} - \Ex \bracket*{\Pest(e_1)\mid \mc{E}_1, \mc{E}_2}} &= \abs*{\Ex \bracket*{\Pest(u) \cdot (1-\Pest(v))\mid \mc{E}_1} - \Ex \bracket*{\Pest(u) \cdot (1-\Pest(v))\mid \mc{E}_1, \mc{E}_2}} \\
        &= |\Cov \paren*{\Pest(u),(1-\Pest(v)) \mid \mc{E}_1} + \Ex \bracket*{\Pest(u)\mid \mc{E}_1}\Ex \bracket*{(1-\Pest(v))\mid \mc{E}_1} \\
        &\hspace{1cm}- \Cov \paren*{\Pest(u),(1-\Pest(v)) \mid \mc{E}_1, \mc{E}_2} - \Ex \bracket*{\Pest(u)\mid \mc{E}_1, \mc{E}_2}\Ex \bracket*{(1-\Pest(v))\mid \mc{E}_1, \mc{E}_2}| \\
        &= \abs*{\Cov \paren*{\Pest(u),(1-\Pest(v)) \mid \mc{E}_1}} + \abs*{\Cov \paren*{\Pest(u),(1-\Pest(v)) \mid \mc{E}_1, \mc{E}_2}} \\
        &+ \abs*{\Ex \bracket*{\Pest(u)\mid \mc{E}_1}\Ex \bracket*{(1-\Pest(v))\mid \mc{E}_1}-\Ex \bracket*{\Pest(u)\mid \mc{E}_1, \mc{E}_2}\Ex \bracket*{(1-\Pest(v))\mid \mc{E}_1, \mc{E}_2}}.
    \end{aligned}
    \]

    We can now bound each of the three terms and begin with the two covariances. The bound on the first follows the same logic as the bound from the second which is larger, so we just use the second bound on both terms for simplicity. We use \Cref{lem:cauchy-schwarz} to see that $\Cov \paren*{\Pest(u),(1-\Pest(v)) \mid \mc{E}_1, \mc{E}_2} \leq \sqrt{\Var\paren*{\Pest(u) \mid \mc{E}_1, \mc{E}_2}\Var\paren*{\Pest(v) \mid \mc{E}_1, \mc{E}_2}}$. From \Cref{lem:variance+mean-shift}, we can see that each variance is at most $\sigma_k^2 \cdot g(3\Tmax) \leq 4\sigma_k^2$. This is because the conditioning $\mc{E}_1,\mc{E}_2$ is made up of the intersection of the events of the trees of 4 vertices being sampled in $W$, the two vertices of each edge. Thus, since two of these vertices are $u$ and $v$, the total number of other vertices who's sampling is conditioned on is at most $3\Tmax$ for each vertex. This gives us the following. 
    \begin{equation}\label{eq:var-vertex-of-edge}
        \begin{aligned}
            \Var\paren*{\Pest(u) \mid \mc{E}_1, \mc{E}_2} &\leq 4\sigma_k^2, \\
            \Var\paren*{\Pest(v) \mid \mc{E}_1, \mc{E}_2} &\leq 4\sigma_k^2, \\
            \Cov\paren*{\Pest(u),\Pest(v) \mid \mc{E}_1, \mc{E}_2} &\leq 4\sigma_k^2. \\
        \end{aligned}
    \end{equation}

    We also bound the difference term similarly, we use the fact that $|ab-a'b'| \leq a|b-b'|+b'|a-a'|$ and the fact that all expectations are in $[0,1]$ to see that we must only bound the differences $\abs*{\Ex\bracket*{\Pest(u) \mid \mc{E}_1} - \Ex\bracket*{\Pest(u) \mid \mc{E}_1, \mc{E}_2}}$ and $\abs*{\Ex\bracket*{\Pest(v) \mid \mc{E}_1} - \Ex\bracket*{\Pest(v) \mid \mc{E}_1, \mc{E}_2}}$. By \Cref{lem:variance+mean-shift} again as well as similar splitting of the events $\mc{E}_1,\mc{E}_2$ into their constituent events, we can see that these terms are at most $3\sigma_k \cdot f(3\Tmax) \leq 6\sigma_k$. Thus, we can attain a bound on $\abs*{\Ex \bracket*{\Pest(e_1)\mid \mc{E}_1} - \Ex \bracket*{\Pest(e_1)\mid \mc{E}_1, \mc{E}_2}}$.
    \[
    \abs*{\Ex \bracket*{\Pest(e_1)\mid \mc{E}_1} - \Ex \bracket*{\Pest(e_1)\mid \mc{E}_1, \mc{E}_2}} \leq 4\sigma_k^2 + 4\sigma_k^2 + 6\sigma_k + 6\sigma_k \leq 13\sigma_k.
    \]

    Now, referring back to \Cref{eq:var-cutval-1}, and the fact stated before that $|ab-a'b'| \leq a|b-b'|+b'|a-a'|$ we bound $\abs*{\Ex \bracket*{\Pest(e_1)\mid \mc{E}_1}\Ex \bracket*{\Pest(e_2)\mid \mc{E}_2} - \Ex \bracket*{\Pest(e_1)\mid \mc{E}_1, \mc{E}_2}\Ex \bracket*{\Pest(e_2)\mid \mc{E}_1, \mc{E}_2}}$ by $13\sigma_k + 13\sigma_k = 26\sigma_k$. Thus, it holds that
    \[
    \Cov \paren*{X_{e_1}, X_{e_2}} \leq \Cov \paren*{\Pest(e_1),\Pest(e_2) \mid \mc{E}_1, \mc{E}_2} + 26\sigma_k.
    \]

    Now, we use \Cref{lem:cauchy-schwarz}, to bound this covariance, We see that $\Cov \paren*{\Pest(e_1),\Pest(e_2) \mid \mc{E}_1, \mc{E}_2} \leq \sqrt{\Var \paren*{\Pest(e_1) \mid \mc{E}_1, \mc{E}_2}\Var \paren*{\Pest(e_2) \mid \mc{E}_1, \mc{E}_2}}$. Consider the first as the bound for the second is the same, $\Pest(e_1) = \Pest(u) \cdot (1-\Pest(v))$. Thus, by \Cref{eq:var-vertex-of-edge} and \Cref{lem:var-product}, we have that $\Var \paren*{\Pest(e_1) \mid \mc{E}_1, \mc{E}_2} \leq 36\sigma_k^2 \leq \sigma_k$. Thus, we can bound our covariance by
    \begin{equation}\label{eq:var-cutval-good-cov-bound}
        \Cov \paren*{X_{e_1}, X_{e_2}} \leq \sigma_k + 26\sigma_k = 27\sigma_k.
    \end{equation}

    Thus, the total contribution to the sum of covariances in \Cref{eq:var-cutval-start} from good pairs is at most $|\C|^2 \cdot 27\sigma_k$. Next, we bound the total contribution of bad pairs. We see that each bad pair has covariance at most $n^{4c\Tmax}$ (as indicator variables for edges are in the range $[0,n^{2c\Tmax}]$), so we bound the total number of them. From \Cref{lem:count-tree-containment}, each vertex $u$ is contained in the tree of at most $n^{q\cdot2^k}$ other vertices. Thus, each edge $(u,v)$ is a bad pair with at most $2n^{q\cdot2^k}$ other edges. Thus, the total contribution of bad pairs to the sum of covariances is at most $|\C|\cdot 2n^{q\cdot2^k}\cdot n^{4c\Tmax}$. Now, we refer back to \Cref{eq:var-cutval-start} and conclude the proof as follows.
    \[
    \begin{aligned}
        \Var \paren*{\cutval} &\leq \frac{1}{|\C|^2} \cdot \paren*{|\C|^2 \cdot 27\sigma_k + |\C|\cdot 2n^{q\cdot2^k}\cdot n^{4c\Tmax}} \\
        &\leq 27\sigma_k + 2n^{q\cdot2^k-4c\Tmax-1+c} \leq 28\sigma_k \leq \sigma_{k+1}^2.
    \end{aligned}
    \]
\end{proof}

\begin{lemma}\label{lem:cutval-valG-Pesth}
    \[
    \Pr(\abs*{\cutval - \val{G'}{\Pesth}} > 2\delta_{k+1}) \leq \delta_{k+1}
    \]
\end{lemma}

\begin{proof}
    From \Cref{lem:var-cutval} and \Cref{lem:chebyshev}, we have
    \[
    \Pr(\abs*{\cutval - \Ex[\cutval]} > \delta_{k+1}) \leq \frac{\sigma_{k+1}^2}{\delta_{k+1}^2} \leq \delta_{k+1}.
    \]
    Thus, the lemma follows from a triangle inequality with this and \Cref{lem:mean-cutval}.
\end{proof}


\section{Final Proof of Correctness}

Finally, we have everything we need to prove the correctness and accuracy of the streaming algorithm (\Cref{alg:streaming-alg}). We restate \Cref{thm:main} more precisely and prove it.

\begin{lemma}
    There is a randomized, adversarial streaming algorithm using $\O{n^{1-c}}$ space that runs on a graph $G$ and returns a value $\Out$ such that
    \[
    \paren*{\frac{1}{2}-\epsilon} \cdot \maxval{G} \leq \Out \leq \maxval{G},
    \]
    with probability greater than $1-\epsilon$.
\end{lemma}

\begin{proof}

    From \Cref{lem:streaming-reduction}, it suffices to show that we have an algorithm that produces a $\paren*{\frac{1}{2}-17\epsilon^2}$-approximation with failure probability at most $\epsilon^2$ on a graph $G'$ with assumptions \ref{assumption:linearedges}, \ref{assumption:balanced-colors}, and \ref{assumption:balanced-in-out}.

    Thus, it suffices to show that our streaming algorithm outputs a value $\Out$ with
    \begin{equation}\label{eq:final-suffices}
        \Pr\paren*{\paren*{\frac{1}{2}-17\epsilon^2}\cdot\maxval{G'} \leq \Out \leq \maxval{G'}} \geq 1-\epsilon^2.
    \end{equation}

    \begin{claim}\label{claim:Out-valG-pos}
        It holds that
        \[
        \Pr\paren*{\abs*{\Out-\paren*{1-16\epsilon^2}\cdot \val{G'}{\pos}} \leq \epsilon^2} \geq 1-\epsilon^2.
        \]
    \end{claim}
    \begin{proof}
        To prove this, we use a series of randomness fixings and triangle inequalities corresponding to how we performed our analysis of the algorithm. We first fixed the randomness of $\B$ which allowed us to prove the inequalities in \Cref{eq:high-deg-estimates} that were necessary in the following section. This failed with probability at most $\delta_0$. Then, we proved \Cref{lem:valG-Pesth-valG-pos} which states
        \[
        \Pr(\abs*{\val{G'}{\Pesth} - \val{G'}{\pos}} > 13\delta_k) \leq \delta_k.
        \]
        This allowed us to fix the randomness in $\Rsvsin(\cdot), \Rsvsout(\cdot)$ with failure probability $\delta_k$. After fixing this, we proved \Cref{lem:cutval-valG-Pesth} which states
        \[
        \Pr(\abs*{\cutval - \val{G'}{\Pesth}} > 2\delta_{k+1}) \leq \delta_{k+1}.
        \]
        Thus, because we have decoupled the randomness in each fixing, we can combine them with a triangle inequality and a union bound. Also noting our probability of failing because of the space bound of $2\epsilon^3$, our total failure probability is at most $\delta_0 + \delta_k + \delta_{k+1} +2\epsilon^3\leq \epsilon^2$. Conditioning on success, we show that the bound holds, proving the claim.
        \begin{align*}
            \abs*{\Out-\paren*{1-16\epsilon^2}\cdot \val{G'}{\pos}} &= \abs*{\paren*{1-16\epsilon^2}\cdot\cutval-\paren*{1-16\epsilon^2}\cdot \val{G'}{\pos}} \\
            &\leq \paren*{1-16\epsilon^2}\cdot\paren*{13\delta_k + 2\delta_{k+1}} \leq \epsilon^2. \qedhere
        \end{align*}
    \end{proof}

    From this, in the case where the bound in \Cref{claim:Out-valG-pos} holds, we have
    \[
    \paren*{1-16\epsilon^2}\cdot \val{G'}{\pos} - \epsilon^2 \leq \Out \leq \paren*{1-16\epsilon^2}\cdot \val{G'}{\pos} + \epsilon^2.
    \]
    Using \Cref{lem:half-apx-valGpos} and the fact that $\val{G'}{\pos} \geq (1/2-\alpha)\cdot \maxval{G'} \geq 1/16$, it holds that
    \[
    \begin{aligned}
        \paren*{1-16\epsilon^2}\cdot \val{G'}{\pos} - 16\epsilon^2\cdot\val{G'}{\pos} \leq &\Out \leq \paren*{1-16\epsilon^2}\cdot \val{G'}{\pos} + 16\epsilon^2\cdot\val{G'}{\pos},\\
        \paren*{1-32\epsilon^2}\cdot \val{G'}{\pos} \leq &\Out \leq \val{G'}{\pos},\\
        \paren*{\frac{1}{2}-17\epsilon^2}\cdot\maxval{G'} \leq &\Out \leq \maxval{G'}.
    \end{aligned}
    \]
    As this happens with probability at least $1-\epsilon^2$, this concludes the proof as we have shown \Cref{eq:final-suffices}.
\end{proof}


\section*{Acknowledgments}
We thank the anonymous STOC’26 reviewers for their helpful comments and suggestions.

\bibliographystyle{plainnat}
\bibliography{references}

\appendix
\section{General Lemmas About Statistics}

\begin{lemma}\label{lem:clamp-basic}
    For any variables $x$ and $y$ in $\mathbb{R}$,
    \[
    \abs{\clamp{x}-\clamp{y}} \leq \abs{x-y}.
    \]
\end{lemma}

\begin{proof}
    We prove this by cases. Case 1 is where $x$ and $y$ are both in the range $[0,1]$. This follows trivially as $\clamp{x} = x$ and similarly for $y$. In case 2 let $x > 1$ and $y \leq 1$. Then $\clamp{x} = 1$, so $\abs{\clamp{x}-\clamp{y}} \leq \abs{1-y} \leq \abs{x-y}$. Lastly in case 3 let $x > 1$ and $y > 1$. Then, $\clamp{x}=\clamp{y}=1$, so the claim is trivial. The other cases follow symmetrically to either case 2 or case 3.
\end{proof}

\begin{proposition}[Chebyshev's Inequality]\label{lem:chebyshev}
    For any random variable $X$ with mean $\mu$ and variance $\sigma^2$ and constant $t > 0$, it holds that
    \[
    \Pr \paren*{\abs*{X-\mu} \geq t} \leq \frac{\sigma^2}{t^2}.
    \]
\end{proposition}

\begin{proposition}[Multiplicative Chernoff Bound]\label{lem:mult-chernoff}
Let $X_1, \dots, X_n$ be independent random variables taking values in $[0,1]$, 
and let $X = \sum_{i=1}^n X_i$ with $\mu = \mathbb{E}[X]$. 
For any $0 < \delta < 1$,
\[
\Pr \paren*{\abs{X - \mu} \geq \delta \mu} \leq 2 \exp\paren*{-\frac{\delta^2}{3} \mu}
\]
\end{proposition}

\begin{proposition}[Additive Chernoff]\label{lem:add-chernoff}
Let $X_1,\dots,X_N$ be independent Bernoulli random variables and set
$X:=\sum_{i=1}^N X_i$ with $\mu:=\mathbb{E}[X]$. Then for any $\lambda>0$,
\[
\Pr\!\left[\,|X-\mu|\ge \lambda\,\right]
\;\le\; 2\exp\!\left(-\frac{\lambda^2}{3\mu}\right).
\]
(Equivalently, for one-sided deviations:
$\Pr[X-\mu\ge \lambda]\le \exp\!\big(-\lambda^2/(3\mu)\big)$ and
$\Pr[\mu-X\ge \lambda]\le \exp\!\big(-\lambda^2/(3\mu)\big)$.)
\end{proposition}

\begin{proposition}[Additive Hoeffding Bound] \label{lem:hoeffding}
Let $X_1, \dots, X_n$ be independent random variables taking values in $[0,1]$, 
and let $X = \sum_{i=1}^n X_i$ with $\mu = \mathbb{E}[X]$. 
For any $t > 0$,
\[
\Pr \paren*{\abs{X - \mu} \geq t} \leq 2 \exp\paren*{-\frac{2t^2}{n}}.
\]
\end{proposition}

\begin{lemma}\label{lem:frac-diff-prelim}
    For positive variables $a, b, c, d$, the following inequality holds:
    \[
    \abs*{\frac{a}{b}-\frac{c}{d}} \leq \frac{\abs*{a-c} \cdot d + \abs*{b-d} \cdot c}{bd}.
    \]
\end{lemma}

\begin{proof}
    \[
    \abs*{\frac{a}{b}-\frac{c}{d}} = \abs*{\frac{ad-cd+cd-bc}{bd}} \leq \frac{\abs*{a-c} \cdot d + \abs*{b-d} \cdot c}{bd}. \qedhere
    \]
\end{proof}

\begin{lemma}\label{lem:custom-prob-ex}
    Let $X$ be a random variable in the range $[0,1]$, let $0 \leq \mu, \delta \leq 1$ such that $\Pr \paren*{\abs*{X-\mu} > \delta} \leq \delta$. It holds that
    \[
    \abs*{\Ex \bracket*{X}-\mu} \leq 2\delta.
    \]
\end{lemma}

\begin{proof}
    Let $\mc{E}$ be the event where $\abs*{X-\mu} > \delta$. Then $\Pr(\mc{E}) \leq \delta$. We can split the expectation by chain rule.
    \[
    \Ex[X] = \Pr\paren*{\mc{E}}\Ex\bracket*{X|\mc{E}} + \Pr\paren*{\overline{\mc{E}}} \Ex\bracket*{X|\overline{\mc{E}}}
    \]
    We see that the left term is in $[0,\delta]$, and the right term is in $[(1-\delta)\cdot(\mu-\delta),\mu+\delta] \subset [\mu-2\delta,\mu+\delta]$. Thus, $\Ex[X] \in [\mu-2\delta,\mu+2\delta]$ proving the claim.
\end{proof}

\begin{proposition}[Jensen's Inequality] \label{lem:jensen}
    For a convex function $f$ and random variable $X$, $f(\Ex[X]) \leq \Ex[f(X)]$. We note that $f(x) = x^2$ and $f(x) = \abs*{x}$ are examples of convex functions.
\end{proposition}

\begin{lemma}\label{lem:root-variance}
    Let $X$ be a random variable. It holds that
    \[
    \Ex \bracket*{\abs*{X-\Ex \bracket*{X}}} \leq \sqrt{\Var(X)}.
    \]
\end{lemma}

\begin{proof}
    From \Cref{lem:jensen}, we have that
    \[
    \Ex \bracket*{\abs*{X-\Ex \bracket*{X}}}^2 \leq \Ex \bracket*{\paren*{X-\Ex \bracket*{X}}^2} = \Var(X).
    \]
    Thus, the lemma follows from taking the square root of both sides.
\end{proof}

\begin{proposition}[Cauchy-Schwarz Inequalities] \label{lem:cauchy-schwarz}
    For any two random variables $X$ and $Y$ it holds $\Cov(X, Y)^2 \leq \Var(X) \Var(Y)$ and $\Ex\bracket*{XY}^2 \leq \Ex\bracket*{X^2}\Ex\bracket*{Y^2}$.
\end{proposition}

\begin{lemma} \label{lem:sampling-covariance}
    Let $X_1, \ldots, X_n$ be a set of (possibly correlated) random variables, such that $X_i \in [0, 1]$ and $\Var(X_i) \leq \sigma^2$. Let $I, J \in [n]$ be independent uniform random variables.
    It holds:
    $$
    \Cov(X_I, X_J) \leq \sigma^2.
    $$
\end{lemma}
\begin{proof}
    We have:
    $$\Cov(X_I, X_J) = \Ex[\Cov(X_I, X_J \mid I, J)] + \Cov(\Ex[X_I \mid I, J], \Ex[X_J \mid I, J]).$$
    The first term on the right-hand side is at most $\sigma^2$ by \cref{lem:cauchy-schwarz} and the assumption that $\Var(X_i) \leq \sigma^2$ for any fixed $i$:
    $$
    \Ex[\Cov(X_I, X_J \mid I, J)]
    \leq \Ex\left[ \sqrt{\Var(X_I \mid I, J)} \sqrt{\Var(X_J \mid I, J)}\right]
    \leq \Ex[\sigma^2]
    = \sigma^2.
    $$
    The second term on the right-hand side is zero, since $\Ex[X_I \mid I, J] = \Ex[X_I \mid I]$ is a function of $I$, and similarly $\Ex[X_J \mid I, J]$ is a function of $J$.
    Therefore, the two are independent, since $I$ and $J$ are independent.
    This concludes the proof.
\end{proof}

\begin{lemma}\label{lem:conditional-clamping}
    Let $X, Y$ be a random variables. Then, it holds:
    $$
    \card{\Ex[\altclamp(X)] - \Ex[\altclamp(Y)]} \leq 
    \card{\Ex[X] - \Ex[Y]} + \sqrt{\Var(X)} + \sqrt{\Var(Y)}. 
    $$
\end{lemma}

\begin{proof}
    We refer to the left-hand side of the inequality as $\mathsf{(LHS)}$, We have:
    \begin{align}
        \mathsf{(LHS)} &\leq \abs*{\Ex [\clamp{X}] - \clamp{\Ex[X]}} + \abs*{\clamp{\Ex[X]} - \clamp{\Ex[Y]}} + \abs*{\Ex [\clamp{Y}] - \clamp{\Ex[Y]}} \notag \\
        &\leq \Ex \bracket*{\abs*{\clamp{X} - \clamp{\Ex[X]}}} 
        + \abs*{\clamp{\Ex[X]}-\clamp{\Ex[Y]}} + \Ex \bracket*{\abs*{\clamp{Y} - \clamp{\Ex[Y]}}} \tag{\Cref{lem:jensen}} \\
        &\leq \Ex \bracket*{\abs*{X - \Ex[X]}} 
        + \abs*{\Ex[X]-\Ex[Y]} + \Ex \bracket*{\abs*{Y - \Ex[Y]}} \tag{\Cref{lem:clamp-basic}}\\
        &\leq \sqrt{\Var(X)} + \abs*{\Ex[X]-\Ex[Y]} + \sqrt{\Var(Y)} \tag{\Cref{lem:root-variance}}.
    \end{align}
    This concludes the proof.
\end{proof}

\begin{lemma}\label{lem:var-of-clamp}
    For random variable $X$, it holds that
    \[
    \Var\paren*{\clamp{X}} \leq \Var(X).
    \]
\end{lemma}

\begin{proof}
    First, for any random variable $Y$ and independent copy $Y'$, it is true that $\frac{1}{2}\Ex[(Y-Y')^2] = \frac{1}{2}\paren*{2\Ex[Y^2]-2\Ex[Y]^2} = \Var(Y)$. Applying this to $Y = \clamp{X}$ and copy $Y' = \clamp{X'}$ and using \Cref{lem:clamp-basic} gives
    \[
    \Var\paren*{\clamp{X}} = \frac{1}{2} \Ex \bracket*{\paren*{\clamp{X}-\clamp{X'}}^2} \leq \Ex \bracket*{\paren*{X-X'}^2} = \Var(X). \qedhere
    \]
\end{proof}

\begin{lemma}\label{lem:cond-Ex-implies-Ex}
    For a random variable $X$ and constants $\mu, \delta, \epsilon$ taking values in $[0,1]$, and an event $A$ which occurs with $\Pr(A) > 1-\delta$. Let $\abs*{\Ex [X|A] - a} \leq \epsilon$. It holds that
    \[
    \abs*{\Ex[X] - a} \leq \epsilon+\delta.
    \]
\end{lemma}

\begin{proof}
    We prove this by using the chain rule on $\Ex[X]$. 
    \[
    \Ex[X] = \Pr(A)\Ex[X|A] + \Pr(\overline{A})\Ex[X|\overline{A}].
    \]
    Now, we examine the intervals of each term based on the products of the bounds on their components. The term on the left must be in the interval $[(1-\delta)\cdot(a-\epsilon),1\cdot(a+\epsilon)] \subset [a-\epsilon-\delta, a+\epsilon]$. Similarly, the right term must be in the range $[0,\delta]$. This gives an overall range for $\Ex[X]$ of $[a-\epsilon-\delta, a+\epsilon+\delta]$ proving the claim.
\end{proof}

\begin{lemma}\label{lem:var-product}
    For random variables $X,Y$ taking values in $[0,1]$, with $\Var(X) \leq \epsilon$, $\Var(Y) \leq \epsilon$, it holds that
    \[
    \Var(XY) \leq 9\epsilon.
    \]
\end{lemma}

\begin{proof}
    Let $\mu_X = \Ex[X]$, $\mu_Y = \Ex[Y]$, and let $X' = X - \mu_X$, $Y' = Y - \mu_Y$. It holds that
    \[
    XY - \mu_X\mu_Y = X'Y' + X'\mu_Y + Y'\mu_X.
    \]
    Thus, we can express the variance of $XY$ as follows, using the fact that $(a+b+c)^2 \leq 3(a^2+b^2+c^2)$.
    \[
    \begin{aligned}
        \Var(XY) &\leq \Ex \bracket*{(X'Y' + X'\mu_Y + Y'\mu_X)^2} \\
        &\leq 3\paren*{\Ex \bracket*{(X'Y')^2} + \Ex \bracket*{(X'\mu_Y)^2} + \Ex \bracket*{(Y'\mu_X)^2}}.
    \end{aligned}
    \]
    We now notice that $\Ex[X'^2] = \Var(X) \leq \epsilon$ and similarly for $Y$, In addition, we can use \Cref{lem:cauchy-schwarz} to see that $\Ex[X'^2Y'^2] \leq \sqrt{\Ex[X'^4]\Ex[Y'^4]}$. Then, since $X'$ takes values in $[0,1]$, we can see that $\Ex[X'^4] \leq \Ex[X'^2] \leq \epsilon$, and similarly for $Y'$. Thus, we can put this together with the inequality above to conclude that $\Var(XY) \leq 3(\epsilon+\epsilon+\epsilon) \leq 9\epsilon$.
\end{proof}
\end{document}